\documentclass[runningheads]{llncs}

\usepackage{amsmath, nccmath} 
\usepackage[ruled,vlined]{algorithm2e} 
\usepackage{amssymb} 
\usepackage{listings}
\usepackage{enumitem}
\usepackage{tikz}
\usepackage{mathtools}
\usetikzlibrary{arrows, matrix, positioning}
\usepackage{color}
\newcounter{mynotes}
\def\B{\hat{B}}
\def\X{\hat{X}}
\def\T{\hat{T}}
\def\S{\hat{S}}
\def\C{\hat{C}}
\def\D{\hat{D}}
\def\A{\hat{A}}
\newtheorem{fact}{Fact}
\newtheorem{Claim}{Claim}
\newcommand{\remove}[1]{}

\DeclareMathOperator*{\argmin}{arg\!\min}

\usepackage{tabularx,lipsum,environ}
\makeatletter
\newcommand{\problemtitle}[1]{\gdef\@problemtitle{#1}}
\newcommand{\probleminput}[1]{\gdef\@probleminput{#1}}
\newcommand{\problemquestion}[1]{\gdef\@problemquestion{#1}}
\NewEnviron{myproblem}{
  \problemtitle{}\probleminput{}\problemquestion{}
  \BODY
  \par\addvspace{.5\baselineskip}
  \noindent
  \begin{tabularx}{\textwidth}{@{\hspace{\parindent}} l X c}
    \multicolumn{2}{@{\hspace{\parindent}}l}{\textsc{\@problemtitle}} \\
    \textbf{Input:} & \@probleminput \\
    \textbf{Question:} & \@problemquestion
  \end{tabularx}
  \par\addvspace{.5\baselineskip}
}

\begin{document}

\title{The Tandem Duplication Distance Problem is hard over bounded alphabets} 

\titlerunning{Tandem Duplication Distance} 

\author{Ferdinando Cicalese\inst{1}
 \and
Nicol\`o Pilati\inst{1}}

\authorrunning{Cicalese, Pilati}

\institute{Department of Computer Science, University of Verona, Italy
\email{\{ferdinando.cicalese,nicolo.pilati\}@univr.it}}
\maketitle              

\begin{abstract}
A tandem duplication is an operation that converts a string $S = AXB$ into a string $T = AXXB.$
As they appear to be involved in genetic disorders, 
tandem duplications are widely studied in computational biology. Also, 
tandem duplication mechanisms have been recently studied in different contexts, from formal languages, 
to information theory, to error-correcting codes for DNA storage systems. 
The question of determining the  complexity of computing the tandem duplication distance between 
two given strings was  posed by  [Leupold {\em et al.},  2004] and, very recently, the problem was  shown 
to be NP-hard  for the case of unbounded alphabets [Lafond et al., STACS 2020]. 

In this paper, we significantly improve this result and show that the tandem duplication distance problem
 is NP-hard  already for the case of strings over an alphabet of size $\leq 5.$ 

For a restricted class of strings, we establish the tractability of the {\em existence problem}: given strings $S$ and $T$ over the same alphabet, 
decide whether there exists a sequence of 
duplications converting $S$ into $T$. A polynomial time algorithm that solves this existence
problem was only known for the case of the binary alphabet. 
\end{abstract}
\section{Introduction}
Since the draft sequence of the human genome was published, it has been known that a very large part 
of it consists of repeated substrings \cite{Lander2001}. 
One talks about a tandem repeat 
when a pattern of one or more nucleotides occurs twice and the two occurrences are adjacent. 
For instance, in  the word $CTACTAGTCA,$ 
the substring $CTACTA$ is a tandem repeat. 
In this case, we say that $CTACTAGTCA$ is generated from $CTAGTCA$ by a tandem duplication of length three. 
As tandem repeats appear to be correlated to several genetic disorders \cite{Sutherland3636,Usdin2008},
the study of tandem duplication mechanisms has attracted the interest of different communities also outside the 
specific area of computational biology \cite{Leupold2004,AlonBFJ16,Farnoud2016,Jain2017Codes,Chee2019}. 

\medskip

\noindent
{\bf Problem definition.}
Formally, a \textit{tandem duplication} (TD)---later simply referred to as a duplication---is an operation 
on a string $S$ that replaces a substring $X$  with its {\em square} $XX.$ 
Given two strings $S$ and $T,$ we write $S \Rightarrow T$ if there exist strings $A$, $B$, $X$ such that $S=AXB$ and $T=AXXB$.

The \textit{tandem duplication distance from $S$ to $T$}, 
denoted by $dist_{TD}(S,T)$, 
is the minimum value of $k$ such that $S$ can be turned into $T$ by a sequence of $k$ duplications. If no such $k$ exists, 
then $dist(S, T) = \infty$.

For example, $dist_{TD}(0121, 0101211)=2$, since: (i) we have the two duplications   
$0121 \Rightarrow 010121 \Rightarrow 0101211$; (ii) and it is easy to verify that no single duplication can 
turn $0121$ into $0101211.$ 

We consider the following two  problems 
about the possibility of converting a string $S$ into a string $T$ by  using 
tandem duplications:

\begin{myproblem}
\problemtitle{Tandem Duplication Existence  (TD-Exist)}
 \probleminput{Strings $S$ and $T$ over the same alphabet $\Sigma$.}
 \problemquestion{Is $dist_{TD}(S, T) < \infty$?}
\end{myproblem}

\begin{myproblem}
\problemtitle{Tandem Duplication Distance (TD-Dist)}
\probleminput{Strings $S$ and $T$ over the same alphabet $\Sigma$ and an integer $k$.}
\problemquestion{Is $dist_{TD}(S, T) \leq k$?}
\end{myproblem}

Determining the complexity of computing the tandem duplication distance between two given strings (the {\sc TD-Dist} problem) was  
posed in \cite{Leupold2004}. Only very recently, the problem was shown 
to be NP-hard in the case of unbounded alphabets \cite{Lafond}. 
Here, we significantly improve this result by showing that {\sc TD-Dist} is NP-hard  for the case of bounded alphabets of size $\leq 5.$  


%

For both, the result of \cite{Lafond} and ours, it is assumed that the strings  $S$ and $T$ satisfy $dist_{TD}(S, T) < \infty.$ 
In general, the complexity of deciding if a string $S$ can be turned into a string $T$ by a sequence of tandem duplications
(the {\sc TD-Exist} problem above) is still an open problem for alphabets of size $> 2.$
In the second part of the paper we also consider this {\em existence problem} ({\sc TD-Exist}), 
focussing  on 
a special class of strings---which we call {\em purely alternating} (see Section \ref{sec:preliminary} for the definition)---that generalize the special structure of binary strings to larger alphabets. We show that a linear time algorithm for the {\sc TD-Exist} problem 
exists for every alphabet of size $\leq 5$ if the strings are {\em purely alternating}.
In a final section we also discuss the limit of the approach used here for larger alphabets $|\Sigma| > 5$.

\medskip

\noindent
{\bf Related Work.}
To the best of our knowledge, the first papers explicitly dealing with tandem duplication mechanisms are  in the area of formal languages 
\cite{Ehrenfeucht1984,Bovet1992,Leupold2004,Leupold2005301,ItoLS-DLT06,LeupoldDLT-07,LeupoldTCS07}. 
In  \cite{Leupold2005301,Leupold2004}, the authors study decidability and hierarchy issues of 
a duplication language, defined  as the set of words generated 
via tandem duplications. In the same line of research, 
Jain et al. \cite{Jain2017}  proved that $k$-bounded duplication languages are regular for $k \leq 3.$
More recently, the authors of  \cite{AlonBFJ16} investigated  
extremal and information theoretic questions regarding  the number of tandem duplications required to 
generate a binary word starting from its unique root (the square free word from which it can be generated via duplications). 
In the same paper, the authors also considered approximate duplication operations. 
In \cite{Farnoud2016} Farnoud et al. began the study of  the average information content of a $k$-bounded duplication language, 
referred to as the capacity of a duplication system.   
Motivated by problems arising from DNA storage applications, Jain et al. \cite{Jain2017Codes} proposed the study of codes that correct tandem 
duplications to improve the reliability of data storage, and gave optimal constructions for the case where tandem duplication length is at most two. In \cite{Chee2019}, Chee et al.  investigated the question of 
confusability under duplications, i.e., whether, given words $x$ and $y,$ there exists a word $z$ such that $x$ and $y$ can be transformed into
$z$ via duplications. 
They show that even for small duplication lengths, the solutions to this question are nontrivial, and exact solutions are provided for 
the case of tandem duplications of size at most three.

\section{Notation and basic properties} \label{sec:preliminary}
We  follow the terminology from \cite{AlonBFJ16,Lafond}. 
 For any positive integer $n$ we denote by $[n]$ the set of the first $n$ positive integers $\{1, 2, \dots, n\}.$ We also use
$[n]_0$ to denote $[n] \cup \{0\}.$ 
Given a  string $S$ we denote by $\Sigma(S)$ the alphabet of the string $S$, i.e.,  
the set of characters occurring in $S$. If $|\Sigma(S)| = q$ we say that the string is \textit{q-ary}. 
We say that $S'$ is a \textit{subsequence} of $S$, if $S'$ can be obtained by 
deleting zero or more characters from $S$, and we  denote this by $S' \subseteq S.$
$T$ is a {\em substring}  of $S$ if $T$ is a subsequence of $S$ whose characters occur contiguously in $S$. If the substring $T$ occurs 
at the beginning (resp.\ end) of $S$ that it is also called 
a \textit{prefix} (resp. \textit{suffix}) of $S$. .

A string $XX$ consisting of the concatenation of two identical strings is called a \textit{square}.
A string is \textit{square-free} if it does not contain any substring which is a square.  We say that a string $S$ is {\em exemplar} if 
no two characters of $S$ are equal.
Given two strings $S$ and $T,$ we say that there is a {\em duplication} turning $S$ into $T,$ denoted by $S \Rightarrow T,$ 
if there exist strings $A$, $B$, $X$ such that $S=AXB$ and $T=AXXB$.
We use $\Rightarrow_k$ to denote the existence of a sequence of $k$ TD's, i.e., 
$S\Rightarrow_k T$ if there exist 
$S_1,\dots,S_{k-1}$ such that $S \Rightarrow S_1 \Rightarrow \dots \Rightarrow S_{k-1} \Rightarrow T$. We write 
$S \Rightarrow_* T$ if there exists some $k$ such that $S \Rightarrow_k T.$ In this case, we also say that
$T$ is a {\em descendant} of $S$ and $S$ is an {\em ancestor} of $T.$ If $S$ is a square free ancestor of $T$, we also say that 
$S$ is a \textit{root} of $B.$ 

The reverse operation of a tandem duplication consists of taking a square $XX$ in $S$ and deleting one of the two occurences of $X.$ 
This operation is referred to as a \textit{contraction}. We write $ T\rightarrowtail S$ if there exist strings 
$A$, $B$, $X$ such that $T=AXXB$ and $S=AXB$. For a sequence of $k$ contractions, we write  
$T\rightarrowtail_k S$ and for a sequence of an arbitrary number of contractions we write $T \rightarrowtail_* S.$ 
In particular, we have that  $T\rightarrowtail_k S$ if and only if $S\Rightarrow_k T$ and $T\rightarrowtail_* S$ if and only if $S\Rightarrow_* T$.


For any string $A = a_1 \dots a_n$ and  set of indices $I = \{i_1, \dots, i_k\},$ such that $1 \leq i_1 < i_2 < \cdots < i_k \leq n$ we define $dup(A, I)$ 
as the string obtained by duplicating character $a_{i_j}$ for each $j =1,\dots,k,$ i.e., 
$$dup(A, I) = a^{(1)} a_{i_1} a_{i_1} a^{(2)} a_{i_2} a_{i_2} a^{(3)} \cdots a^{(k)} a_{i_k} a_{i_k} a^{(k+1)},$$ 
where $a^{(j)} = a_{i_{j-1}+1} a_{i_{j-1}+2} \dots a_{i_{j}-1},$ and $i_0 = 0, i_{k+1} = n+1.$

For a character $a$ and a positive integer $l,$ let $a^l$ denote the string consisting of $l$ copies of $a.$
A \textit{run} in a string $S$ is a maximal substring of $S$ consisting of copies of the same character. 
Given a string $S$ containing $k$ runs, the \textit{run-length encoding} of $S$, 
denoted $RLE(S)$, is the sequence $(s_1, l_1), \dots, (s_k, l_k),$ such that $S = s_1^{l_1} s_2^{l_2} \dots s_k^{l_k},$ where, in particular, 
for each $i=1,\dots, k,$ we have that  $s_i^{l_i}$ is the $i$-th run of $S$, consisting of the symbol $s_i$ repeated $l_i$ times. 
We write $|RLE(S)|$ for the number of runs contained in $S$. 
For example, given the string $S=111001222 = 1^30^21^12^3,$ we have $RLE(S)=(1,3), (0,2), (1,1), (2,3)$ and $|RLE(S)|=4$.

A $q$-ary string $S = s_1^{l_1} \cdots s_k^{l_k}$ is called \textit{purely alternating} if, there is an order on the symbols of the alphabet 
$\Sigma(S) = \{\sigma_0 < \sigma_1 < \dots \sigma_{q-1}\}$ and a $j \in [k],$ such that for each $i=1, \dots, k$ 
$s_i = \sigma_{j+i \bmod q},$ i.e., each run of the simbol $\sigma_j$ is followed by a run of the symbol $\sigma_{j+1 \bmod q}.$
Note that a purely alternating string is uniquely determined 
by the order on the alphabet, the initial character and the lengths of its runs. In general, we will assume, w.l.o.g., that, for a $q$-ary purely alternating string the alphabet is the set 
$\{0, 1, \dots, q-1\}$ with the natural order and the first run is a run of $0$'s. 

For example, the string $0001220112$ is purely alternating, but the string $01202$ is not. Note that all binary strings (that up to relabelling we assume to start with $0$) are purely alternating. 

Given  a $q$-ary purely alternating string $S$, a \textit{group} of $S$ is a substring $X$ of $S$ containing exactly $q$ runs of $S$.

%
 Given a string $S= s_1 s_2 \dots s_n$, we denote by $S^{dup}$ the string obtained by duplicating each single character in $S$, i.e., 
 $S^{dup} = dup(S, [|S|]) =  s_1 s_1 s_2 s_2 \dots s_n s_n.$

A string $S$ is \textit{almost square-free} if there exists a square-free string 
$S_{SF}$ such that $S_{SF} \subseteq S \subseteq S_{SF}^{dup}.$ 

For example, the string $01120022$ is almost square-free, while the strings 
$01122201$ and $0012212$ are not.

The following lemma records a useful immediate consequence of the definition of an almost square free string.

\begin{lemma} \label{lemma:almsqfreecontr}
$Z$ is an almost square free string if and only if there exists a square free string $Z_{SF}$ 
and a set $I \subseteq [|Z_{SF}|]$ such that 
$Z = dup(Z_{SF}, I).$ Moreover, the only contractions possible  
on $Z$ are of size $1$, i.e., those that remove one of two consecutive equal characters.
\end{lemma}

\section{The Hardness of {\sc TD-Dist} over alphabets of size $5$}

\label{Chapter3}

In this section we show that given two strings $S$ and $T$, over a $5$-ary alphabet $\Sigma$ and such that $S \Rightarrow_* T,$ 
 finding the minimum number of duplications required to transform $S$ into $T$ is NP-complete. 
%

To see that the problem is in NP, note that, if $S \Rightarrow_* T$, then  
$dist_{TD}(S, T) \leq |T|$ because each duplication from $S$ to $T$ adds at least one character. 
Thus a certificate for the problem consisting of the sequence of duplications turning $S$ into $T$ 
has polynomial size and can also be clearly verified in  polynomial time.
 



\bigskip

\noindent
{\bf The Block-Exemplar Tandem Duplication Distance problem.}
For the hardness proof, we reduce from a variant of the
{\sc Tandem Duplication Distance} problem over unbounded alphabets, which we define below.
In this variant, only instances made of strings of a special structure are allowed. We call the allowed 
instances {\em block exemplar pairs}. We provide an operative definition of how a block-exemplar pair of strings is 
built. Recall that a string is exemplar if its characters are all distinct.

\begin{definition}[Block-Exemplar Pair] \label{def:pair}
Fix and integer $t > 1$ and  $t+2$ exemplar strings $X, B_0, B_1, \dots, B_t$ over pairwise disjoint alphabets, such that 
$|X|, |B_0|, |B_1| > 1$ and $|B_i| = 1$ for each $i=2, \dots, p.$ 
Fix a second positive integer $p < t$ and $p$ distinct subsets $\emptyset \neq I_1, \dots, I_p \subsetneq [|X|].$ 
Finally, fix a character $\text{\L}$ which does not appear in any of the 
strings $X, B_0, \dots, B_t.$

The strings $\text{\L}, X, B_0, B_1, \dots, B_t$ and the sets $I_1, \dots, I_p,$ 
determine a block-exemplar pair of strings $S$ and $T$ as follows:

For each $j=1, \dots, p,$ define the distinct strings $X_1, \dots, X_p,$ where $X_j$ 
is obtained from $X$ by single character duplications,   
$X_j = dup(X, I_j),$ hence  $X \subsetneq X_i \subsetneq X^{dup}.$ 
Define the strings
$$ {\cal B}_t^{0}  = B_t B_{t-1} \dots B_2 B_1 B_0^{dup} 
\qquad {\cal B}_t^{1} = B_t B_{t-1} \dots B_2 B_1^{dup} B_0,$$
and for each $i=1, \dots, t,$ define the strings 
$${\cal B}_i = B_i B_{i-1} \dots B_2 B_1 B_0 \qquad 
{\cal B}_i^{01} = B_i B_{i-1} \dots B_2 B_1^{dup} B_0^{dup}.$$
Then set 
\begin{eqnarray}
S &=& {\cal B}_{t} X \text{\L} = B_{t} B_{t-1} \dots B_2 B_1 B_0 X \text{\L}, \label{S-structure}\\
T &=& {\cal B}^0_{t} X^{dup} \text{\L} \,  {\cal B}^1_{t} X \text{\L} \,\,\,
{\cal B}^{01}_{1} X_1 \text{\L} \, {\cal B}^1_{t} X \text{\L} \,\,\,
{\cal B}^{01}_{2} X_2 \text{\L} \, {\cal B}^1_{t} X \text{\L} \,\,\,
\dots 
{\cal B}^{01}_{p} X_p \text{\L} {\cal B}^1_{t} X \text{\L}. \label{T-structure}
\end{eqnarray}

More generally, we say that $S,T$ is a block-exemplar pair, if there exist integers $t, p,$ strings $\text{\L}, X, B_0, \dots, B_t,$ and 
sets $I_1, \dots, I_p$ such that $S$ and $T$ are given by the above construction.
\end{definition}

\noindent
We can now define the following variant of the {\sc TD-Dist} problem.

\begin{myproblem}
\problemtitle{Block-Exemplar Tandem Duplication Distance Problem (B-Ex-TD)}
\probleminput{ A Block-Exemplar pair of strings $S$ and $T,$ an integer bound $k$.}
\problemquestion{Is $dist_{TD}(S, T) \leq k$?}
\end{myproblem}

In \cite{Lafond} the authors showed how to map instances $(G,k)$ of  the  NP-complete problem {\sc Clique} 
to choices of strings $X, B_0, \dots, B_t,$ and sets $I_1, \dots, I_p$ and a parameter $k'$ such that $G$ has  
clique of size $k$ if and only if $S \Rightarrow_{k'} T,$ where $S,T$ is the block-exemplar pair given by   
$X, B_0, \dots, B_t, I_1, \dots, I_p.$ Hence we have that the following result is implicit in \cite{Lafond}.

\begin{theorem}  \label{theo:cost-effective-hardness}
The \textsc{B-Ex-TD} problem is NP-complete.
\end{theorem}

\medskip

\noindent
{\bf From {\sc B-Ex-TD} to {\sc TD-Dist} on alphabets of size 5.}
\remove{
We will prove the following theorem. 
\begin{theorem}
\label{thm:k5tdnpcomplete}
The {\sc TD-Dist} problem is NP-complete over alphabets of size $5$.
\end{theorem}


\medskip

For the hardness proof, we reduce from the 
\textsc{Block-Exemplar Tandem Duplication Distance} problem. 
We will show that for every {\em block-exemplar pair of strings} $S, T$ there is a pair of 
strings $\hat{S}, \hat{T}$ over an alphabet of size $5$ such that 
$dist_{TD}(S, T) \leq k$ if and only if $dist_{TD}(\hat{S}, \hat{T}) \leq k.$
}



The basic idea  is to map a block-exemplar pair $S$ and $T$ to a pair of 
$5$-ary strings $\hat{S}, \hat{T}$ having a structure analogous to the one of $S$ and $T$ (as given in (\ref{S-structure})-(\ref{T-structure})). 
However,  instead of using exemplar strings for the building blocks $X, B_{0}, \dots, B_{t}$ we use  
substrings of a square free ternary string.
The crucial point is to carefully choose such substrings so that we can control the squares appearing  in the resulting $\hat{S}$ and $\hat{T}.$  
We want that the only duplications allowed in any process transforming $\hat{S}$ into $\hat{T}$ can be one-to-one mapped 
to some duplication sequence from  $S$ to $T.$ 
For this, we need to overcome significant technical hurdles, that make the hard part of the reduction non-trivial
(i.e., the equivalence in the direction from the $5$-ary variant considered here to the block-exemplar variant of the problem considered in \cite{Lafond}).

%


\medskip

\noindent
{\bf  The mapping from the block-exemplar pair $(S,T)$ to the $5$-ary pair $(\hat{S}, \hat{T}).$}
We first show how to map the different substrings that define the structure of  the pair $S$ and $T$ (see Def. \ref{def:pair}) 
into $5$-ary strings. Then we will show that the resulting strings $\hat{S}, \hat{T}$ 
preserve the  properties of the original strings with respect to the possible duplication sequences $\hat{S} \Rightarrow_* \hat{T}.$ 

The strings $\hat{S}$ and $\hat{T}$ will be defined on the alphabet $\{0,1,2, \text{\L}, \$\}.$
For each $i=0, 1, \dots, t$ we  define a string $\hat{B}_i \in \{0,1,2, \$\}^*$ that will be used to encode the string $B_i.$
Also we  define the encoding  of strings $X, X_1, \dots, X_p$ into strings $\hat{X}, \hat{X}_1, \dots, \hat{X}_p$
over the alphabet $\{0,1,2\}.$ 

The encodings $\B_i$ and $\X$ are obtained by iteratively "slicing off" suffixes from a long square free string, denoted by $\mathcal{O}.$ In addition 
each  $\B_i$ is also extended with a single occurrence of the character \$.

More precisely, let us define $\mathcal{O}$ to be a square free string of length at least 
$|X| + t^2 + t  (5+ \max\{|B_0|, |B_1|\})$) over the alphabet $\{0, 1,  2\}.$
Note that creating $\mathcal{O}$ takes polynomial time in $|S|+|T|$ using a square-free morphism, for example Leech's 
morphism (\cite{Leech}). 
Let us now proceed to define the  building blocks $\tilde{B}_t, \dots, \tilde{B}_0, \tilde{X}.$ 
Refer to Fig.\ \ref{fig:blocks} for a pictorial description of this process. 
\begin{figure}[htbp] 
   \centering
   \includegraphics[width=\textwidth]{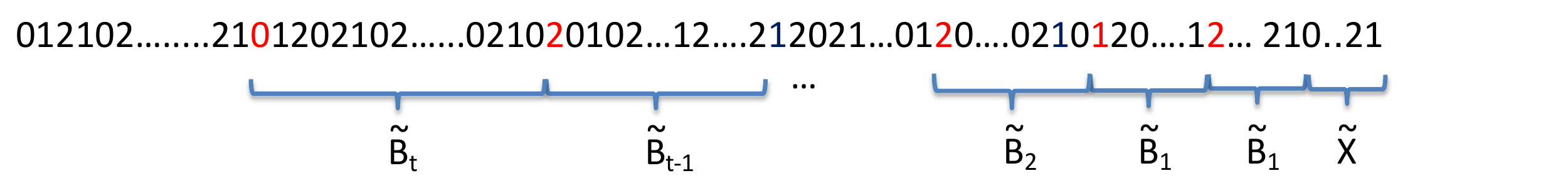} 
   \caption{The way the blocks $\tilde{B}_{t}, \tilde{B}_{t-1}, \dots, \tilde{B}_{1} \tilde{B}_{0}, \tilde{X}$ are 
   consecutively extracted from the square free ternary string ${\cal O}.$} 
   \label{fig:blocks}
\end{figure}
First, we define the string $\tilde{X}$ to be the suffix of length $|X|$ of $\mathcal{O}.$ 
%
Let ${\cal O}^{(0)}$ be the string obtained from ${\cal O}$ after removing the suffix $\tilde{X}.$
Let $\tilde{B}_{0}$ be the shortest suffix of $\mathcal{O}^{(0)}$ that does not start with $0$ and 
has length at least $\max(|B_0|, |B_1|).$ 
We then define $\tilde{B}_{1}$ as the shortest suffix of the string obtained by removing the suffix $\tilde{B}_{0}$ from 
$\mathcal{O}^{(0)}$ such that $|\tilde{B}_{1}|>|\tilde{B}_{0}|$ and 
that doesn't start with 0. 

For $i = 2, \dots t$ we define the string $\tilde{B}_i$ as the minimum suffix of the string ${\cal O}^{(i-1)}$ obtained from 
$\mathcal{O}$ after removing 
the suffix $\tilde{B}_{i-1} \tilde{B}_{i-2} \dots \tilde{B}_1\tilde{B}_0\tilde{X}$ and such that: 
(i) $|\tilde{B}_i|>|\tilde{B}_{i-1}|$; 
(ii) if $i \neq t$ then $\tilde{B}_i$ does not start with 0;
(iii) if $i = t$ then $\tilde{B}_i$ starts with 0.
Thus, the lengths of the strings $\tilde{B}_{0}, \tilde{B}_{1}, \tilde{B}_2, \dots, \tilde{B}_{t}$ 
 is monotonically increasing, all the strings are ternary square-free and only for $j = t$ we have that 
 $\tilde{B}_j$ starts with the character 0. Note that, in order to guarantee that the starting character of the strings $\tilde{B}_i$ is (resp.\ is not) $0,$ it is enough to 
consider a suffix of ${\cal O}^{(i-1)}$ of size at most $|\tilde{B}_{i-1}| + 6$. This follows from ${\cal O}^{(i-1)}$ being a ternary square free string, hence it cannot contain 
substrings of size 5 using only two characters, since every string of size 5 over only two characters necessarily contains a square. 
Therefore, the prefix of length $5$ of the suffix of length $|\tilde{B}_{i-1}|+6$ of ${\cal O}^{(i-1)}$ 
contains an occurrence of any character in $\{0,1,2\}$ that can be
chosen as the desired starting character of $\tilde{B_i}.$ Therefore, the number of character we use for  the strings
$\tilde{X}, \tilde{B}_0, \dots  \tilde{B}_t,$ is at most $|X| + \sum_{i=0}^t \max\{|B_0|, |B_1|\}+i+5 \leq |X| + t^2 + t(5+ \max\{|B_0|, |B_1|\}) = |{\cal O}|.$

We now use these strings and the two additional characters $\text{\L}, \$$ as building blocks to create larger strings which will be our $5$-ary ``analog'' of the 
strings $\mathcal{B}_i.$ 


Let $\Omega$ be the following set of strings:
 $$\Omega = \{\text{\L}\} \cup \{B_2, B_3, \dots, B_{2p}\} \cup \mathbb{B}_0 \cup \mathbb{B}_1 \cup \mathbb{X},$$ 
where 
$\mathbb{B}_0 = \{B'_0 \mid B_0 \subseteq B'_0 \subseteq B_0^{dup}\},$ 
$\mathbb{B}_1 = \{B'_1 \mid B_1 \subseteq B'_1 \subseteq B_1^{dup}\},$
$\mathbb{X} = \{X'  \mid X \subseteq X' \subseteq X^{dup}\}.$
Note that $\Omega$ contains all the almost square free strings 
that either appear among the building blocks of $T$ and $S$ (\ref{S-structure})-(\ref{T-structure}) or can be obtained from 
such building blocks via single letter duplications.
It turns out that (see Fact \ref{fact:contractionsTS-core} below)
any string encountered in a sequence of contractions $T \rightarrowtail_* S$ is a concatenation of 
elements of $\Omega.$ For this reason we refer to the elements of $\Omega$ as the macro-characters or {\em chunks} of $S$ and $T.$


Note that, by Lemma \ref{lemma:almsqfreecontr},
 for any $C' \in \mathbb{B}_0 \cup \mathbb{B}_1 \cup  \mathbb{X},$ there exists a unique $C \in \{B_0, B_1, X\}$ and a unique set $I \subseteq [|C|]$ such that
$C' = dup(C, I).$

We now define an encoding of the elements of $\Omega.$ Let $\mu: A \in \Omega \mapsto \mu(A) \in \{0,1,2,\$, \text{\L}\}^*$ be defined as follows 
(for the sake of highlighting the factors involved in 
the formulas,  $a \cdot b$ will denote concatenation  of strings $a$ and $b$):
\begin{itemize}
\item $\mu(\text{\L}) = \text{\L},$ and $\mu(X) = \tilde{X},$ 
\item for each $i = 2, \dots, t,$ we set $\mu(B_i) = \tilde{B}_i \cdot \$ $
\item for each $B' \in \mathbb{B}_0 \cup \mathbb{B}_1,$ we let $\mu(B') = dup(\mu(B), I) \cdot \$,$ where, by Lemma \ref{lemma:almsqfreecontr}, 
 $B \in \{B_0, B_1\}$ and $I \subseteq [|B|]$ are uniquely defined by $B' = dup(B, I)$;
\item for each $X' \in \mathbb{X}$ we let $\mu(X') = dup(\mu(X), I),$ where, by Lemma \ref{lemma:almsqfreecontr},
  $I \subseteq [|X|]$ is uniquely defined by $X' = dup(X, I)$;
\end{itemize}
The mapping $\mu$ is naturally extended to concatenations of elements of $\Omega$ by setting 
$\mu(A_1 \cdot A_2  \cdots  A_r) = \mu(A_1) \cdot \mu(A_2) \cdots  \mu(A_r),$
for each $A_1, \dots, A_r \in \Omega$

For each $i=1, \dots, t,$ let  $\B_i = \mu(B_i)$ and for $i=0,1,$ let $\B_i^* = \mu(B_i^{dup}).$ Finally, let $\X = \mu(X)$ and $\X^* = \mu(X^{dup}).$ 
The set of macro-characters  $\hat{\Omega}$ that constitute the range of $\mu$ is then given by: 
 $$\hat{\Omega} = \{\text{\L}\} \cup \{\B_2, \B_3 , \dots, \B_{t}  \} \cup \mathbb{\B}_0 \cup \mathbb{\B}_1 \cup \mathbb{\X},$$
 where 
$\mathbb{\B}_0 =  \{\B'_0  \mid \B_0 \subseteq \B'_0 \subseteq \B_0^*\},$ 
$\mathbb{\B}_1 = \cup \{\B'_1  \mid \B_1 \subseteq \B'_1 \subseteq \B_1^*\},$ and 
$ \mathbb{\X} = \{\X' \mid \X \subseteq \X' \subseteq \X^{dup}\}.$
It is  easy to see that  $\mu$ is one-one from $\Omega$ to $\hat{\Omega}.$

We are ready to define the  $5$-ary strings  $\hat{S}$ and $\hat{T}$ that we use in our reduction:
\begin{eqnarray}
\hat{S} &=& \mu(S) =  \B_{t}  \B_{t-1}  \dots  \B_2  \B_1  \B_{0}  \X \text{\L} 
= \tilde{B}_{t} \$ \tilde{B}_{t-1} \$ \dots \$ \tilde{B}_2 \$ \tilde{B}_1 \$ \tilde{B}_{0} \$ \tilde{X} \text{\L}\\
\hat{T} 
  &=& \mu(T) =  \mathcal{\B}_{t}^0 \X^{*} \text{\L} \, \mathcal{\B}_{t}^1 \X \text{\L} \,\, \mathcal{\B}_{1}^{01} \X_{1} \text{\L} \mathcal{\B}_{t}^1 \X \text{\L} \,\,
   \mathcal{\B}_{2}^{01} \X_{2} \text{\L} \mathcal{\B}_{t}^1 \X \text{\L} \, \dots \, \mathcal{\B}_{p}^{01} \X_{p} \text{\L} \mathcal{\B}_{t}^1 \X \text{\L},
\end{eqnarray}
where for each $i=1, \dots, t,$ we let $\X_i = \mu(X_i),$ and for  $q \in [t]$, and $a \in \{0,1,01\}$ we let
$\mathcal{\B}_q = \mu(\mathcal{B}_q),$ and 
$\mathcal{\hat{B}}_q^a = \mu(\mathcal{B}_q^a).$
Note that 
$\hat{S}$ is square-free. In analogy with $S$ and $T$, we refer to the substrings from $\hat{\Omega}$ that constitute the 
building blocks of $\S$ and $\T$ as {\em chunks}.

\medskip
\noindent
{\bf The final steps of the hardness proof.} 
Exploiting the properties of $\mu$ and the structure of $S, T, \S, \T,$ we have the following result.
\begin{theorem} \label{lemma:equivalence}
Let $S, T$ be a block exemplar pair of strings and let $\hat{S}$ and $\hat{T}$ be the corresponding pair of $5$-ary strings built as above. 
Then $S\Rightarrow_k T$ iff $\hat{S}\Rightarrow_k \hat{T}.$
\end{theorem}

The hardness of the {\sc TD-Dist} problem over a $5$-ary alphabet,  immediately follows from Theorem \ref{theo:cost-effective-hardness} 
and Theorem \ref{lemma:equivalence}. \qed

\vspace{-0.5cm}

\subsection{Sketch of the proof of Theorem \ref{lemma:equivalence}: The main analytic tools}

The first properties we use are summarized in the following fact. They come from the interdependence between $\mu$ and the operator $dup$ and their effect on concatenations
 of a prefix of a chunk with a suffix of a chunk. 
They basically say that the equality of the prefixes (resp.\ suffixes) of pairs of chunks from 
$\Omega$ is preserved also by the images of such chunks via the map $\mu.$ 
Moreover, the concatenation of prefixes and suffixes of chunks from $\Omega$ result in a chunk from $\Omega$ iff the same holds true for their images via $\mu.$  
For a string $C=c_1 c_2 \dots c_n$, and indices $i, j \in [n]$ we denote by  $C[i,j]$ the substring $c_i c_{i+1} \dots c_j,$ if $i \leq j$ and the empty string otherwise.

\begin{fact} \label{mu-dup-modularity-2}
Let $C \in \{B_0, B_1, X\}$ and $C \subseteq C' \subseteq C^{dup}.$ 
Let $I$ be the unique set of indices such that 
$C' = dup(C, I).$ Let $\hat{C} = \mu(C)$ and denote by $\hat{C}_h = \hat{C}[1,|C|]$ and $\hat{C}_t = \hat{C}[|C|+1, |\hat{C}|].$ 
Let $c = |C|, c' = |C'|, \hat{c} = |\hat{C}|.$  
Then, it holds that
\begin{enumerate}
\item $\mu(C') = dup(\C, I) = dup(\hat{C}_h, I) \cdot \hat{C}_t;$
\item for any $I', I'' \subseteq [c]$  and $i \in [c+|I'|]_0, \, j \in [c +|I''|]_0,$ we have 
\begin{itemize}
\item $\displaystyle{dup(C, I')[1, i] = dup(C, I'')[1,i]  \iff dup(\C, I')[1,i] = dup(\C, I'')[1, i]}$
\item $\displaystyle{dup(C, I')[i+1, c+|I'|] = dup(C, I'')[j+1,c+|I''|]}$\\
$\displaystyle{~~~~~\iff dup(\C, I')[i+1, \hat{c} +|I'|] = dup(\C, I'')[j+1, \hat{c}+|I''|]}$
\end{itemize}
\remove{ 
\item for any $I', I'', I''' \subseteq [c],$ and $i \in [c+|I'|], j \in [c+|I''|],$ we have 
\begin{itemize}
\item $\displaystyle{dup(C, I')[1,i] \cdot dup(C, I'')[j+1,c+|I''|] = dup(C, I''')}$\\
$\displaystyle{~~~~\Rightarrow
dup(\C, I')[1,i] \cdot dup(\C, I'')[j+1, \hat{c}+|I''|) = dup(\C, I''')}$
\end{itemize}
}
\item for any $I', I'', I''' \subseteq [|C|],$ and $i \in [\hat{c} +|I'|], j \in [\hat{c}+|I''|],$ we have 
\begin{itemize}
\item $\displaystyle{dup(\C, I')[1,i] \cdot dup(\C, I'')[j+1,\hat{c}+|I''|] = dup(\C, I''')}$\\
$\displaystyle{~\Rightarrow
dup(C, I')[1,\min\{i, c+|I'|] \cdot dup(C, I'')[j+1,c+|I''|] = dup(C, I''')}$
\end{itemize}
\end{enumerate}
\end{fact}
The first property directly follows from the definition of $\mu$. 
The following properties are immediate consequences of property 1 and the
 one-one correspondence between 
the position of the duplications in a chunk $C' = dup(C, I)$ and the position of
 the duplication in its image $\mu(C') = dup(\mu(C), I).$ This correspondence 
comes from the fact that both strings $C$ and $\mu(C)$ are square free.

\medskip
\noindent
{\bf Part 1 - The proof that $S\Rightarrow_k T$ implies  $\hat{S}\Rightarrow_k \hat{T}.$}
Assume that $S\Rightarrow_k T$ and 
let $T = T(0) \rightarrowtail T(1) \rightarrowtail \cdots \rightarrowtail T(k) = S$ be the corresponding series of contractions 
leading from $T$ to $S$. 
We recall a result that directly follows from \cite[Appendix, Claim 2]{Lafond}. 

\begin{fact} \label{fact:contractionsTS-core}
For each $\ell = 0, 1, \dots, k-1,$ the string $T(\ell)$ has a factorization  into elements of $\Omega,$ i.e., $T(\ell) = A_1 A_2 \dots, A_r$ where
 for each $i=1, \dots, r, $ we have 
$A_i \in \Omega.$ Moreover, adjacent chunks are over disjoint alphabets, i.e., for each $i=1, \dots, r-1,$  $\Sigma(A_i) \neq \Sigma(A_{i+1}).$
\end{fact}

Exploiting the factorization of $T(\ell)$ into chunks, guaranteed by Fact \ref{fact:contractionsTS-core}, and using  
property 2 in Fact \ref{mu-dup-modularity-2} we have that  the following claim holds. 
\noindent
\begin{Claim} \label{claim:equivalence1-new}
For $\ell = 0, 1, \dots, k-1,$ the contraction $T(\ell) \rightarrowtail T(\ell+1)$ implies a contraction $\mu(T(\ell)) \rightarrowtail \mu(T(\ell+1)).$
\end{Claim}

Applying this claim for each contraction from $T$ to $S,$ we have 
that a sequence of contractions exists from $\T$ to $\S,$ i.e., 
$\T =
 \mu(T(0))  \rightarrowtail \mu(T(1)) \rightarrowtail \cdots \rightarrowtail \mu(T(k)) = \S.$
Hence $\hat{S}\Rightarrow_k \hat{T},$ proving the ``only if''  part of Theorem \ref{lemma:equivalence}.

The proof of Claim \ref{claim:equivalence1-new} is based on the fact that  each contraction $T(\ell) \rightarrowtail T(\ell+1)$ 
is either a single character within a chunk $A_i$ of the
factorization of $T(\ell)$ given by Fact \ref{fact:contractionsTS-core} or it is a contraction that starts in a chunk $A_i$ and spans at least $A_i A_{i+1} A_{i+2} A_{i+3}.$ To see this 
let $DD$ be the square removed by the contraction and let us denote by $D^L$ and $D^R$ the left and right copy of $D$, with the first character of $D^L$ being 
in $A_i.$ 

(i) If the first character of $D^R$ is also in $A_i,$ then $D^L$ is a substring of $A_i.$ Since $\Sigma(A_i) \cap \Sigma(A_{i+1}) = \emptyset,$ 
$D^R,$ being equal to $D^L$ cannot extend to $A_{i+1}$. Since $A_i$ is almost square free, the only square are pairs of single characters, hence 
$D^LD^R$ must be a pair of adjacent single characters. In this case, by 1. in Fact \ref{mu-dup-modularity-2}, 
a single letter square is present in the chunk $\mu(A_i)$ and the contraction that removes it implies $\mu(T(\ell)) \rightarrowtail \mu(T(\ell+1)).$ 

(ii) If the first character of $D^R$ is not in $A_i,$  the first character of $D^L$ cannot be from $A_{i+1}$ which does not have any character occurring in $A_i.$
The analogous argument about the last characters of $D^R$ and $D^L$ implies that $D^LD^R$ spans at least four chunks.
Then, $D^L = A_i[u+1,|C_1|] D' A_j[1, u'],$ i.e., it starts with a suffix $A_i[u+1,|A_i|]$ of chunk $A_i$ 
and ends with a (possibly empty) prefix $A_j[1,u']$ of a chunk $A_j$ with $j \geq i+2$ and $D^R$ starts with the suffix of $A_j[u'+1, |A_j|]$ of $A_j$ and  
ends with the (possibly empty) prefix $A_{j'}[1,u']$ of a chunk $A_{j'}.$ The remaining parts of $D^L$ and $D^R$ must be equal and coincide with some sequence of chunks $D'$, i.e.,  $D^L = A_i[u+1,|A_i|] \cdot D' \cdot A_j[1, u'] = A_j[u'+1,|A_j|] D' A_{j'}[1, u'] = D^R.$ Then, 
by using  2. of Fact \ref{mu-dup-modularity-2}  we show that there is 
an equivalent square $\D^L\D^R$ in $\mu(T(\ell)$ corresponding to 
$D^L = \mu(C_1)[i+1, |\mu(C_1)|] \mu(D') \mu(C_2)[1, j]  = \mu(C_2)[j+1, |\mu(C_2)|] \mu(D') \mu(C_3)[1, j] = \D^R.$ 
Contracting it, implies $\mu(T(\ell)) \rightarrowtail \mu(T(\ell+1))$ also in this case.   The complete proof is in Appendix.

\medskip
\noindent
{\bf Part 2 - The proof that $\S \Rightarrow_k \T$ implies  $S\Rightarrow_k T.$} In this case, 
we show the possibiliy of mapping every sequence of contractions $\T \rightarrowtail \T(1) \rightarrowtail \T(2) \rightarrowtail \cdots \rightarrowtail  \S$ 
 to a sequence of contraction $T \rightarrowtail_k S.$ 
 Finding a characterization of the strings $\T(\ell)$ analogous to Fact \ref{fact:contractionsTS-core} is significantly more involved.
Although by construction $\T$ and $\S$ also have a factorization into chunks of $\hat{\Omega},$ most of these chunks (also adjacent ones) 
are from the same $4$-ary alphabet. Therefore,  showing that every intermediate string $\T(\ell)$ 
is also factorizable into chunks requires more care. 
Moreover, this fact, together with the invertibility of $\mu$, shows only that we can find strings 
$T(0) = \mu^{-1}(\T), T(1) = \mu^{-1}(\T(1)), \dots, T(k-1)= \mu^{-1}(\T(k-1)), \, $ each of which is factorizable into chunks of $\Omega.$ 
Since a contraction $\T(\ell) \rightarrowtail \T(\ell+1)$ can involve suffixes and prefixes of chunks, we need a deeper analysis of  
 such strings to characterize precisely the structure of the possible contractions $\T(\ell) \rightarrowtail \T(\ell+1)$ (Proposition \ref{proposition:form} in appendix). 
 Then, we show that  3. in Fact \ref{mu-dup-modularity-2} guarantees the existence of a contraction $T(\ell) \rightarrowtail T(\ell+1).$ 
Due to the space limitations, this part of the proof is deferred to the appendix. 

\vspace{-0.3cm}

\section{TD-distance for purely alternating strings} 
\label{Chapter2} 

\vspace{-0.2cm}

In this section, we investigate the existence of polynomial time algorithms to decide whether a purely alternating string $S$ can be transformed into another purely alternating string $T$ through a  series of duplications, i.e., if $S \Rightarrow_* T$.

\begin{definition}
Fix strings  $S= s_1^{l_1} s_2^{l_2} \dots s_n^{l_n} $ and $T = t_1^{l'_1} t_2^{l'_2} \dots t_m^{l'_m}.$
We say that the run $s_i^{l_i}$ \textit{matches} the run $t_j^{l'_j}$ if  $s_i = t_j$ and $l_i \leq l'_j$. We also say that $S$ matches $T$ 
(and write $S \preceq T.$) if $n=m$ and for  $i = 1, \dots, n$ we have  $s_i^{l_i}$ matches $t_j^{l'_j}$. 
\end{definition}
The existence of a string $S'$ that matches $T$ and that satisfies $S \Rightarrow_* S'$, implies that $S \Rightarrow T$:   
we can convert $S'$ into $T$ by duplications on single letters. 
\begin{definition} \label{defi:normal-duplications}
Given two $q$-ary strings $S$ and $T$, we say that the operation $S=AXB \Rightarrow AXXB=T$ is a \textit{normal duplication} if 
one of the following conditions holds: (i) $X$ is a $q$-ary string with exactly $q$ runs such that the first and the last run are of length 1; 
(ii) $X$ is a single character.
\end{definition}
We write $S \Rightarrow^N T$ if there exists a normal duplication converting $S$ into $T$. We write $S\Rightarrow^N_k T$ if there exist $S_1,\dots,S_{k-1}$ such that $S \Rightarrow^N S_1 \Rightarrow^N \dots \Rightarrow^N S_{k-1} \Rightarrow^N T$. 
We  write $S \Rightarrow^N_* T$ if there exists some $k$ such that 
$S \Rightarrow^N_k T$. 

In perfect analogy with the definition of contractions given in section \ref{sec:preliminary}, we  define {\em normal contractions}: 
$T\rightarrowtail^N_k S$ and $T \rightarrowtail^N_* S$  by  $T\rightarrowtail^N_k S$ if and only if $S\Rightarrow^N_k T$ and $T\rightarrowtail^N_* S$ if and only if $S\Rightarrow^N_* T$.


Intuitively, normal duplications are effective in converting a string $S$ into a string $S'$ 
that matches a string $T$ because they keep the resulting string purely alternating and create new runs that are as small as possible; 
these runs allow the string $S'$ to match many strings. 


We proceed to characterizing pairs of strings $S$ and $T$ such that $S \Rightarrow^N_* T.$ 

\begin{lemma} \label{lemma:function}
Fix $2\leq q \leq 5.$ Let $S = s_0^{l_1} s_1^{l_2} \dots s_n^{l_n}$ and $T = t_0^{l'_1} t_1^{l'_2} \dots t_m^{l'_m}$ 
be purely alternating strings over the same $q$-ary alphabet $\Sigma = \{0,1,2,\dots, q-1\}.$ 
 Then, 
$S \Rightarrow^N_* T$ if and only if there exists 
a function $f:\{1, \dots, n-q+2\}\mapsto\{1, \dots, m-q+2\}$ 
such that: (1.) $f(1) = 1$ and $f(n-q+2) = m-q+2$; 
(2.) $f(i) = j \implies s_i = t_j$ and  for each $u=0, \dots, q-2$ we have that $l_{i+u} \leq l'_{j+u}$;
(3.)  $f(i) = j$ and $f(i') = j'$ and $i<i' \implies j<j'$;
(4.) if $q = 5$ and $f(i) = j$ and $f(i+1) = j' \neq j+1 \implies$ there exists a substring 
    $M$ in $T$ starting in a position $p$ such that $j \leq p \leq j'$ 
    with the form $M=s_{i+3}^{l'_p}, s_{i+4}^{l'_{p+1}}, \dots  s_{i+q}^{l'_{p+q-3}}, s_{i+q+1}^{l'_{p+q-2}}$ 
    such that for each $u=0, 1, \dots, q-3$ it holds that $l_{i+3+u} \leq l'_{p+u}$ and $l_{i+1} \leq l'_{p+q-2}.$
\end{lemma}

It turns out that for for alphabets of size $\leq 5$ the set of normal duplications is as ``expressive'' as the set of all  
possible duplications. We use the following claim (see the appendix for the proof).

\begin{Claim}\label{claim:narymodule}
Let $S$ and $T$ be $q$-ary purely alternating strings such that
  $S \Rightarrow_* T$. Let $AXA' \Rightarrow AXXA'$ be one of the duplications of the sequence leading from $S$ to $T$. Then,  $|RLE(X)| \bmod q \leq 1$.
\end{Claim}

\begin{lemma} \label{lemma:normal-equivalence}
Let $S$ and $T$ be purely alternating strings over the same alphabet $\Sigma$ of size $\leq 5$. Then 
$S \Rightarrow_* T$ if and only if 
$S \Rightarrow^N_* T.$  
\end{lemma}
\begin{proof} 
Using Claim \ref{claim:narymodule}, we show that any duplication 
can be simulated by 
normal duplications. We give the complete argument only for 
for the case $|\Sigma| = 5.$  The cases $|\Sigma| \in \{2,3,4\}$ can be showed analogously. 


\begin{Claim}
\label{claim:quinarynormal}
Let $S$ and $T$ be a $5$-ary purely alternating strings over the alphabet $\Sigma = \{0,1,2,3,4\}.$ If there exists a duplication $S \Rightarrow T$, then we can create a series of normal duplications $S \Rightarrow \dots \Rightarrow T'$ such that $T'$ matches $T$.
\end{Claim}

\begin{proof}
Let $S = AXB \Rightarrow AXXB = T$ be the original duplication. Like before, we create a string that matches $XX$ starting from $X$ through normal duplications, depending on how many runs are contained in $X$. By Claim \ref{claim:narymodule} the only possible cases are $|RLE(X)| \bmod 5 \in \{0,1\}$

\noindent
{\em Case 1.} $|RLE(X)| = 1.$ It means that the only effect of the original duplication is to extend one of the runs of $S$. For this reason  $S$ must already match $T$.

\noindent
{\em Case 2.} $|RLE(X)| \bmod 5 = 0$, we suppose that the string $X$ starts with a $0$ (rotate the characters if it starts with any other symbol), 
so $X$ has the form $X = 0^{l_1} 1^{l_2} 2^{l_3} 0^{l_4}\dots 1^{l_{n-1}} 2^{l_n}$. If $|RLE(X)|$ is equal to $5$ then it is aleady  a normal duplication.

If $|RLE(X)| = 10$, we consider the following two sequences of normal duplications to match $XX$ (the duplicated part is underlined): 
\remove{
\begin{align*}
 X ={}& \underline{0^{l_1} \ 1^{l_2} \ 2^{l_3} \ 3^{l_4} \ 4^{l_5}} \ 0^{l_6} \ 1^{l_7} \ 2^{l_8} \ 3^{l_9} \ 4^{l_{10}} \\
\Rightarrow{}& 0^{l_1} \ 1^{l_2} \ \underline{2^{l_3} \ 3^{l_4} \ 4^{1} \ 0^{1} \ 1^{l_2}} \ 2^{l_3} \ 3^{l_4} \ 4^{l_5} \ 0^{l_6} \ 1^{l_7} \ 2^{l_8} \ 3^{l_9} \ 4^{l_{10}} \\
\Rightarrow{}& 0^{l_1} \ 1^{l_2} \ 2^{l_3} \ 3^{l_4} \ 4^{1} \ 0^{1} \ 1^{1} \ 2^{1} \ \boxed{3^{l_4}} \ 4^{1} \ 0^{1} \ 1^{l_2} \ 2^{l_3} \ 3^{l_4} \ 4^{l_5} \ 0^{l_6} \ 1^{l_7} \ 2^{l_8} \ 3^{l_9} \ 4^{l_{10}} = X'\\
XX ={}& 0^{l_1} \ 1^{l_2} \ 2^{l_3} \ 3^{l_4} \ 4^{l_5} \ 0^{l_6} \ 1^{l_7} \ 2^{l_8} \ \boxed{3^{l_9}} \ 4^{l_{10}} \ 0^{l_1} \ 1^{l_2} \ 2^{l_3} \ 3^{l_4} \ 4^{l_5} \ 0^{l_6} \ 1^{l_7} \ 2^{l_8} \ 3^{l_9} \ 4^{l_{10}}
\end{align*}
}
\begin{align*}
(i)  \; \; X ={}& \underline{0^{l_1} 1^{l_2}  2^{l_3}  3^{l_4}  4^{l_5}}  0^{l_6}  1^{l_7}  2^{l_8}  3^{l_9}  4^{l_{10}} 
\Rightarrow 0^{l_1}  1^{l_2}  \underline{2^{l_3}  3^{l_4}  4^{1}  0^{1}  1^{l_2}} 2^{l_3}  3^{l_4}  4^{l_5}  0^{l_6}  1^{l_7}  2^{l_8}  3^{l_9}  4^{l_{10}} \\
\Rightarrow{}& 0^{l_1}  1^{l_2} 2^{l_3} 3^{l_4} 4^{1} 0^{1} 1^{1} 2^{1}  \boxed{3^{l_4}} 4^{1} 0^{1} 1^{l_2} 2^{l_3} 3^{l_4} 4^{l_5} 0^{l_6} 1^{l_7}  2^{l_8}  3^{l_9} 4^{l_{10}} = X'\\
(ii) \; \; X ={}& 0^{l_1}  1^{l_2} 2^{l_3} 3^{l_4} 4^{l_5} \underline{0^{l_6} 1^{l_7}  2^{l_8}  3^{l_9}  4^{l_{10}}} 
\Rightarrow 0^{l_1}  1^{l_2}  2^{l_3} 3^{l_4} 4^{l_5} 0^{l_6}  1^{l_7}  \underline{2^{l_8}  3^{l_9}  4^{1} 0^{1} 1^{l_7}}  2^{l_8}  3^{l_9}  4^{l_{10}} \\
\Rightarrow{}& 0^{l_1}  1^{l_2}  2^{l_3}  3^{l_4}  4^{l_5}  0^{l_6}  1^{l_7}  2^{l_8}  3^{l_9}  4^{1}  0^{1}  1^{1}  2^{1}  \boxed{3^{l_9}}  4^{1}  0^{1}  1^{l_7}  2^{l_8}  3^{l_9}  4^{l_{10}} = X''\\
XX ={}& 0^{l_1} 1^{l_2} 2^{l_3}  3^{l_4}  4^{l_5}  0^{l_6}  1^{l_7}  2^{l_8}  \boxed{3^{l_9}}  4^{l_{10}}  0^{l_1}  1^{l_2}  2^{l_3}  3^{l_4}  4^{l_5}  0^{l_6}  1^{l_7}  2^{l_8}  3^{l_9}  4^{l_{10}}
\end{align*}
It is easy to see that according to whether $l_4 \leq l_9$ or not, we have that either $X'$ or $X''$ matches $XX.$ Hence, in either case, we have the desired sequence of duplications proving the claim. 

Finally, let us assume that $X$ contains $5r$ runs for some $r > 2$. Then, in order to produce a 
 string through normal duplications that matches $XX$, it suffices to execute the duplications explained before, then continue with $|RLE(X)|/5-2$ normal duplications containing the four adjacent runs of length $1$ plus another adjacent run. This pushes to the right the original runs of $X$ remaining, together with $\boxed{3^{l_4}}$ (in the first case) or $\boxed{3^{l_9}}$ (in the second case). We can see that $\boxed{3^{l_4}}$ in $X'$ will be in the same position as $\boxed{3^{l_9}}$ in $XX$ and vice versa. 
\end{proof}
\end{proof}
The previous two lemmas imply the following
\begin{theorem} \label{theo:main-purely}
Let $\Sigma$ be an alphabet of size $\leq 5$. There exists a  algorithm that for every pair of 
purely alternating strings $S$ and $T$ over $\Sigma$ can decide in linear time whether $S \Rightarrow_* T.$
\end{theorem}
\begin{proof}
The algorithm computes the run length encoding of $S$ and $T$ and then decides about the existence of the 
function $f$ satisfying the properties of Lemma \ref{lemma:function}.
By Lemma \ref{lemma:normal-equivalence} we have that $S \Rightarrow_* T$ {\em if and only if} 
$S \Rightarrow^N_* T.$ By Lemma \ref{lemma:function} this latter condition holds {\em if and only if} there exists 
a function $f$ satisfying the conditions 1-4 in Lemma \ref{lemma:function}. 
Therefore, to prove the claim it is enough to show that the existence of such a function $f$ can be decided in linear time. 
This is easily attained by employing the following greedy approach (see, e.g.,  Algorithm \ref{alg:narylinear} in appendix):
once the values of $f(1)=1, \dots, f(i-1)= j$ have been fixed, sets the assignment 
$f(i) = j'$ to the smallest $j'$ such that $l_{i+u} \leq l'_{j'+u}$ for each $u=0, \dots  k-1$ and 
if this condition does not hold for $j'=j+1,$ then $j'$ is the smallest integer $> j$  that guarantees the existence of a 
$j < p < j'$ satisfying condition 
$4$ in Lemma \ref{lemma:function}. 
The correctness of this approach can be easily shown by a standard exchange argument 
and it is deferred to the appendix for the sake of the space limitations. The resulting algorithm takes $O(|S|+|T|)$ time, 
since it only scans for a constant number of times each component of the run length encoding of $T$ and $S.$
%
%
\end{proof}
\noindent
{\bf Final remarks on $6$-ary strings and some open problems.}
The technique we used to prove Lemmas \ref{lemma:function}, \ref{lemma:normal-equivalence}
 is not generalizable to the case of larger alphabets: For purely alternating strings with $|\Sigma(S)|=6$, we cannot always simulate a general duplication with normal duplications. 
 Take, e.g.,  the duplication $AXB \Rightarrow AXXB$ where $X$ is the string: 
 $\displaystyle{X = 0^{2} 1^{1} 2^{2} 3^{1} 4^{2} 5^{1} 0^{1} 1^{2} 2^{1} 3^{2} 4^{1} 5^{2}}.$
%
%
One can  show that no pair of normal duplications can  convert the string X into a string $T' \preceq XX.$
%
%
%
We do not know whether this has an implication on the polynomial time solvability  of the problem {\sc TD-Exist} already on $6$-ary alphabets, and we leave it as a first step for future research.  
More generally, the main algorithmic problems that are left open by our results are the complexity of {\sc TS-Dist} for binary alphabets (more generally, whether our hardness result can be extended to smaller alphabets)  and the complexity of {\sc TD-Exist} for arbitrary ternary alphabets. Also, on the basis of the hardness result, approximation algorithms for the distance problem is another interesting direction for future research.


\vspace{-0.3cm}

\bibliographystyle{splncs04}
\bibliography{bibliography}

\newpage

\noindent
{\bf \huge Appendix}

\appendix

\section{Proof of Theorem  \ref{lemma:equivalence}}

\noindent
{\bf Theorem \ref{lemma:equivalence}.}
{\em Let $S, T$ be a block exemplar pair of strings and let $\hat{S}$ and $\hat{T}$ be the corresponding pair of $5$-ary strings built as described above. 
Then $S\Rightarrow_k T$ if and only if $\hat{S}\Rightarrow_k \hat{T}.$
}

\subsection{The ``only if'' part - The proof of Claim \ref{claim:equivalence1-new}}

\noindent
{\bf Claim \ref{claim:equivalence1-new}.} 
{\em 
For $\ell = 0, 1, \dots, k-1,$ the contraction $T(\ell) \rightarrowtail T(\ell+1)$ implies a contraction $\mu(T(\ell)) \rightarrowtail \mu(T(\ell+1)).$
}
\par
\begin{proof} 
Fix $0 \leq \ell \leq k-1.$ Let $T(\ell) = A D D A' \rightarrowtail A D A' = T(\ell+1).$
 We argue by cases, according to whether $D$ is a substring of a single chunk in the factorization of $T(\ell)$ or $D$ spans over more 
 than one chunk.

\noindent
{\em Case 1.}  $D$ is a substring of some chunk $C'$. Then, since adjacent chunks have disjoint alphabet, the whole square $DD$ must be contained 
in the same chunk. Moreover, since $C'$ is almost square free, $D$ is a single letter. Therefore,  there must exist, 
two possibly empty sequences of chunks, $T'$ and $T''$ together with a chuck $C \in \{B_0, B_1, X\},$ a set of indices $I \subseteq [|C|]$ and an index $i^*$ 
such that $C' = dup(C, I \cup \{i^*\}), \, C'' = dup(C, I)$ and the contraction $T(\ell) \rightarrowtail T(\ell+1)$ can be written as: 
$$T(\ell) = T' C' T'' = T' dup(C, I \cup \{i^*\}) T'' \rightarrowtail T' dup(C, I) T'' = T' C'' T'' = T(\ell+1).$$

By definition, $\mu(C') = dup(\mu(C), I \cup \{i^*\})$ and $\mu(C'') = dup(\mu(C), I),$  then  we also have 
the contraction from $\mu(T(\ell))$ to $\mu(T(\ell+1)),$ given by 
\begin{multline*}
\mu(T(\ell)) = \mu(T') \mu(C') \mu(T'')  = \mu(T') dup(\mu(C), I \cup \{i^*\}) \mu(T'') \\
\rightarrowtail \mu(T') dup(\mu(C'), I) \mu(T'')  = 
 \mu(T') \mu(C'')  \mu(T'') = \mu(T(\ell+1)).
\end{multline*}

\medskip
 \noindent
{\em Case 2.} $D$ spans over more than one chunk. Since adjacent chunks have disjoint alphabet, it cannot be that the two copies of $D$ 
are substrings of consecutive chunks. Therefore, if the leftmost character of the left copy of $D$ is in chunk $A_i$ the leftmost character of  the right copy of $D$
must be in chunk $A_{i+2}.$

We first consider the subcase where, for some $1 \leq i < j \leq r-2,$ it holds that  $D = A_{i} A_{i+1} \dots A_{j},$ i.e., $D$  coincides with 
a sequence of whole chunks. Since, by assumption, $T(\ell) = A D D A'$ is a sequence of chunks, then, also $A$ and $A'$ are sequences of chunks, and we have
$\mu(T(\ell)) = \mu(A) \mu(D) \mu(D) \mu(A')$ and $\mu(T(\ell+1))= \mu(A) \mu(D) \mu(A')$ which shows that 
there is a contraction from $\mu(T(\ell))$ to $\mu(T(\ell+1))$ which acts on the repeat $\mu(D) \mu(D).$ 

\medskip
It remains to consider the case where $D$ does not coincide with a sequence of chunks and spans at least two chunks. Then, at least one of the two copies of 
$D$ does not coincide with the first character of the chunks where it starts. Hence, such a chunk is not a single character, i.e., its alphabet is one of 
$\Sigma(B_0), \Sigma(B_1), \Sigma(X).$ The same must hold true for the prefix of the other copy of $D.$

We conclude that there is $C \in \{B_0, B_1, X\},$  sets of  indices $I_1, I_2, I_3 \subseteq [|C|]$ 
indices $i \in [|C|+|I_1|] \cup\{0\}, j \in [|C|+|I_2|] \cup\{0\},$ and chunks $C_1 = dup(C, I_1), C_2 = dup(C, I_2), C_3 = dup(C, I_3),$
such that 
$$T(\ell) = T' C_1 D' C_2  D' C_3 T'',$$ where $T', T'', D'$ are sequences of chunks with 
\begin{equation} \label{condC1C2}
C_1[i+1,|C_1|] = C_2[j+1, |C_2|] 
\end{equation}
and 
\begin{equation}\label{condC2C3}
C_2[1,j] = C_3[1,j] 
\end{equation} 
 so that   
 $$D = C_1[i+1, |C_1|] D' C_2[1,j] = C_2[j+1, |C_2|] D' C_3[1, j]$$ and 
 $$DD = C_1[i+1,|C_1|] D' C_2[1,j] C_2[j+1, |C_2|] D' C_3[1,j],$$ is a square.  
Hence the contraction is 
\begin{multline*}
T(\ell) = T' C_1[1,i] D D C_3[j+1, |C_3| T'' \\
\rightarrowtail T' C_1[1,i] D  C_3[j+1, |C_3| T'' = T' C_1 D' C_3 T'' = T(\ell+1).
\end{multline*}

We factorize $\mu(T(\ell))$ as follows
\begin{eqnarray} 
\mu(T(\ell)) &=& \mu(T') \mu(C_1) \mu(D') \mu(C_2)   \mu (D') \mu(C_3) \mu(T'') \\ 
&=& \mu(T') dup(\mu(C), I_1)[1,i] \\
&& ~~~ dup(\mu(C), I_1)[i+1, |\mu(C)| +|I_1|] \mu(D') 
dup(\mu(C), I_2)[1,j] \label{D-left}\\
&& ~~~dup(\mu(C), I_2)[j+1,|\mu(C)|+|I_2|] \mu(D') dup(\mu(C), I_3)[1,j] \label{D-right}\\
&& ~~~dup(\mu(C), I_3)[j+1,|\mu(C)|+|I_3|] \mu(T'').
\end{eqnarray}
By Fact \ref{mu-dup-modularity-2} (item 2.) and the relations in (\ref{condC1C2}), (\ref{condC2C3}), we have 
that the string in (\ref{D-left}) is equal to the string in (\ref{D-right}).
Hence, 
$$\mu(T(\ell))=  \mu(T') \, dup(\mu(C), I_1)[1,i] \,  \hat{D} \, \hat{D}  \, dup(\mu(C), I_3)[j+1,|\mu(C)|+|I_3|] \, \mu(T''),$$
where 
$$\hat{D} = dup(\mu(C), I_1)[i, |\mu(C)| +|I_1|] \mu(D') 
dup(\mu(C), I_3)[1,j] $$
Therefore, there is a contraction 
\begin{eqnarray*} 
\mu(T(\ell)) &\rightarrowtail& \mu(T') \, dup(\mu(C), I_1)[1,i] \,  \hat{D} \, dup(\mu(C), I_3)[j+1,|\mu(C)|+|I_3|] \, \mu(T'')\\
&=&  \mu(T') \, dup(\mu(C), I_1) \, \mu(D') \, dup(\mu(C), I_3) \, \mu(T'')\\
&=&  \mu(T') \, \mu(C_1) \, \mu(D') \, \mu(C_3) \, \mu(T'') = \mu(T(\ell+1)),
\end{eqnarray*}
which settles also the second subcase of Case 2 and completes the proof of Claim \ref{claim:equivalence1-new}.
\end{proof}

\subsection{The proof of ``if'' part of Theorem  \ref{lemma:equivalence}}

Let us recall the definition of $\T$ and $\S:$
\begin{eqnarray}
\hat{S} &=& \mathcal{\B}_{2p} \X \text{\L} = \B_{2p} \$ \B_{2p-1} \$ \dots \$ \B_2 \$ \B_1 \$ \B_{0} \$ \X \text{\L} \label{S}\\
\hat{T} 
  &=&  \mathcal{\B}_{2p}^0 \X^{dup} \text{\L} \, \mathcal{\B}_{2p}^1 \X \text{\L} \,\, \mathcal{\B}_{1}^{01} \X_{1} \text{\L} \mathcal{\B}_{2p}^1 \X \text{\L} \,\,
   \mathcal{\B}_{2}^{01} \X_{2} \text{\L} \mathcal{\B}_{2p}^1 \X \text{\L} \, \dots \, \mathcal{\B}_{p}^{01} \X_{p} \text{\L} \mathcal{\B}_{2p}^1 \X \text{\L}. \label{T}
\end{eqnarray}

Assume that there exists a sequence of $k$  contractions 
$$\T = \mu(T) = \T(0) \rightarrowtail \T(1) \rightarrowtail \T(2) \rightarrowtail \cdots \rightarrowtail \T(k) = \S,$$ where  
$\T(\ell)$ denotes the string obtained from $\T$ after the first $\ell$ contractions have been performed. 
We tacitly assume that in each contraction $\T(\ell) = ADDA' \rightarrowtail ADA' = \T(\ell+1)$, it is the right copy of $D$ which is 
removed. 

%

 A \textit{block} of $\T(\ell)$ is a {\em maximal} substring 
 $P$ of $\T(\ell)$ satisfying the following properties: the last character of $P$ is  \L{}; and this is the only 
 occurrence of \L{} in $P$. Hence, the first character of $P$ is either preceded by \L{} or it is the first 
 character of $\T(\ell)$.

We denote by $\hat{B}\hat{X}\text{\L}$ the first, leftmost block $\mathcal{\B}_{2p}^0 \X \text{\L}$ of $\T$. 
For $i=1, \dots, p,$ we also denote by $E_i$ the block $\mathcal{\B}_{i}^{01} \X_i \text{\L}$ in $\T.$
We let $E_i(\ell)$ be the substring of $\T(\ell)$ formed by 
all the characters that belong to $E_i$. Note that $E_i(\ell)$ is any possibly empty subsequence of $E_i$.

For any $a \in \{0,1, 01\}$ and $j \in [2p]$,
 a block $\mathcal{\B}'\X'\text{\L}$ is called a \textit{$\mathcal{\B}_{j}^a \X \text{\L}$-block} 
 if $\mathcal{\B}_{j} \X \text{\L} \subseteq \mathcal{\B}'\X'\text{\L} \subseteq \mathcal{\B}_{j}^a \X^{dup} \text{\L}$. 
 In other words, $\mathcal{\B}'\X'\text{\L}$ has the same runs of $\mathcal{\B}_{j}^a \X^{dup} \text{\L}$
  in the same order, but some of the duplicated characters may 
  have been contracted into a single character. 
  A \textit{$\mathcal{\B}_{2p}^1 \X \text{\L}$-cluster} 
  is a string obtained by concatenating an arbitrary number of $\mathcal{\B}_{2p}^1 \X \text{\L}$-blocks. 
  We write ($\mathcal{\B}_{2p}^1 \X \text{\L})^*$ to denote a possibly empty 
  $\mathcal{\B}_{j}^1 \X \text{\L}$-cluster.

\begin{proposition}
\label{proposition:form}
For each $\ell = 0, \dots, k$
\begin{enumerate} 
\item $\T(\ell)$ has the form:
\begin{align} \label{standardform}
\B\X\text{\L}\,(\mathcal{\B}_{2p}^1 \X \text{\L})^* E_{i_1}(\ell)(\mathcal{\B}_{2p}^1 \X \text{\L})^* E_{i_2}(\ell)(\mathcal{\B}_{2p}^1 \X \text{\L})^* \dots E_{i_h}(\ell)(\mathcal{\B}_{2p}^1 \X \text{\L})^*
\end{align}
where:
\begin{itemize}
    \item $\B\X\text{\L}$ is a $\mathcal{\B}_{2p}^0 \X \text{\L}$-block
    \item $i_1 < i_2 < \dots < i_h$
    \item each $(\mathcal{\B}_{2p}^1 \X \text{\L})^*$ is a $\mathcal{\B}_{2p}^1 \X \text{\L}$-cluster
    \item for each $j \in \{i_1, \dots, i_h\}$, $E_j(\ell)$ is a $\mathcal{\B}_{j}^{01} \X \text{\L}$-block.
\end{itemize}
\item the contraction $\T(\ell) = A DD A' \rightarrowtail A D A' = \T(\ell+1)$ satisfies  one of the following  possibilities
\begin{enumerate}
\item $D$ is a single character, hence, necessarily one of the characters in some substring $\X' \in \mathbb{\X}$, 
$\B'_0\in \mathbb{\B}_0,$ $\B'_1 \in \mathbb{\B}_1$---these are 
the substrings corresponding to the elements of $\hat{\Omega}$, which appear in the above factorization of $\T(\ell)$ which are not square free;
\item there exists a contraction $\T(\ell) = A \tilde{D}\tilde{D} A' \rightarrowtail A \tilde{D} A' = \T(\ell+1)$  
such that $\tilde{D}$ is a sequence of whole chunks (note that this case includes the possibility $\tilde{D} = D$, i.e., 
already $D$ is a sequence of whole chunks)
\item there are  chunks $\C_1 = \C_1' \C_1''; \, \C_2 =  \C_2' \C_2''; \,  \C_3 = \C_3' \C_3''$ such that  for some $\C \in \{\B_0, \B_1, \X\}$ for each 
$i=1, 2,3,$ it holds that $\C \subseteq \C_i \subseteq \C^*$ and
 $D = C_1'' D' C_2' = C_2'' D' C_3',$ and $C_1'$ is a suffix of $A$ and $C_3''$ is 
a prefix of $A'.$
%
\end{enumerate}
\end{enumerate}
\end{proposition}

\begin{proof} We argue by induction on $\ell.$
For $\ell = 0,$ the statement is true by construction. 
Let $\ell \geq 0$ and assume the claim is true for $\ell,$ i.e.,
$\T(\ell)$ has the structure in (\ref{standardform}). 
Let $D$ be the string that is contracted from $\T(\ell)$ to $\T(\ell+1)$. We argue by cases:
\begin{enumerate}
\item If $D$ does not contain an \L{} character, then $DD$ is entirely contained in a single block. 
But each block is almost square-free, so the only possibility is that $DD=aa$ for some character $a$. 
We can see that a contraction like this cannot remove an entire run of the block, 
so each $\mathcal{B}_{j}^a \mathcal{X} \text{\L}$-block will remain a $\mathcal{B}_{j}^a \mathcal{X} \text{\L}$-block,
and the item 1. of the claim is true for $T(\ell+1)$ too. Moreover, item 2. is also verified as subcase (a).
We have proved the induction step for this case.

\item Assume now that the last character of $D$ is \L{}. Then $DD = D'\text{\L}D'\text{\L}$ for some string $D'$, and removing the second $D'\text{\L}$ half only removes entire blocks of $\T(\ell)$. Since among these blocks there cannot be the first 
block $\B\X\text{\L}$ and since each $E_{i_j}(\ell)$ is itself a block, this contraction also implies that $\T(\ell+1)$
has the form in (\ref{standardform}), proving item 1. Moreover,  item 2. is also verified as 
subcase (b). The induction step is verified also in this case.
\end{enumerate}

It remains to consider the cases where last character of $D$ is not \L{}, but $D$ has at least one \L{} character. It is easy to see that in such a situation, the second condition in the claim statement is guaranteed. For the other conditions, we have four additional cases to consider depending on which blocks the leftmost character and rightmost character removed belong to, i.e, where the right half of $DD$ starts and ends in $\T(\ell)$.

\begin{enumerate}[start=3]

\item The leftmost character removed belongs to the first block $\B\X\text{\L}$. 

We will show that no such a contraction is possible. 
Since we are assuming that $D$ contains the character \L{},  the right $D$ must contain
the  \L{} that occurs at the end of $\B\X\text{\L}$. 
This is the only \L{} contained in $\B\X\text{\L}$ thus the left $D$ cannot contain one, which implies that
such a contraction cannot happen. 
    
\item The leftmost character removed belongs to an $E_{i_j}(\ell)$ substring for some $j \in [h]$ 
and the rightmost character removed belongs to a $\mathcal{\B}_{2p}^1 \X \text{\L}$-cluster. 

We can show that such a contraction is also impossible: 
Since $D$ contains the character \L{}, the right half of $DD$ must contain the \L{} in $E_{i_j}(\ell)$. 

The leftmost character removed cannot be the first character of $E_{i_j}(\ell)$, 
for otherwise the left copy of $D$ would end with a \L{}, 
against the standing hypothesis that  $D$ does not  end with \L{}. Furthermore, since the right $D$ contains an \L{}, also the left $D$ must contain an \L{}. We know that the first character of the right $D$ is not the first character of $E_{i_j}(\ell-1)$, therefore the left $D$ contains at least a character after the \L{}. Let $D'$ (resp. $D''$) be the maximum suffix of the right (resp.\ left) $D$ that does not
contain  \L{}. We can represent the contraction as follows:
    
\begin{align*}
    \T(\ell)=\underbrace{\dots\text{\L}\overbrace{E_{i_j}(\ell}^{D''}}_{D}\underbrace{)(\mathcal{\B}_{2p}^1 \X \text{\L})^* \overbrace{(\mathcal{\B}_{2p}^1}^{\mathclap{D'}}}_{D}\X \text{\L})\dots
\end{align*}
    
Now we observe that in order to have $DD$ be a valid contraction, it must hold that  $D'=D''$, which is impossible since no $E_{i}(\ell)$ string starts with a 0, while $\mathcal{\B}_{2p}^1$ starts with a 0 by construction. 
Therefore, there cannot be a contraction of the this type either.
    
\item The leftmost character removed belongs to a $\mathcal{B}_{2p}^1 \mathcal{X} \text{\L}$-cluster. 
    
 We start by observing that in this case, it is not possible that  the rightmost character removed belongs 
 to a different $(\mathcal{\B}_{2p}^1 \X \text{\L})^*$-cluster. 
 For otherwise, for some $j>0$, the right $D$ must contain $\text{\L}E_j(\ell)$ as a substring, 
 hence, also the left $D$ must contain this substring. 
 But this is impossible because the number of runs of each block $E_j(\ell)$ is different 
 and each $\mathcal{\B}_{2p}^a \X \text{\L}$-block has more runs than $E_{j}(\ell)$ by construction.

Also, in the standing case, we can show that it is not possible that the rightmost character removed belongs to $E_j(\ell)$ 
for some $j>0$. 
Indeed, we know that the only possibility in such a case would be that $E_j(\ell)$ is the block immediately 
after the cluster where the leftmost character removed belongs 
(for an argument analogous to the one used in the previous paragraph).  
Let $D'$ (resp. $D''$) be the maximum suffix of the right (resp. left) $D$ that doesn't contain \L{}. 
Then, the  contraction must have the following form:
    
\begin{align*}
    \T(\ell)=\underbrace{\dots \text{\L} (\overbrace{\mathcal{\B}_{2p}^1}^{D''}}_{D} \underbrace{\X \text{\L})(\mathcal{\B}_{2p}^1 \X \text{\L})^* \overbrace{E_{j}(\ell}^{\mathclap{D'}}}_{D})\dots
\end{align*}
 For $DD$ to be a contraction, we should have that $D'=D''$. However, by construction, $D'$ cannot begin with the character 0, since $E_j(\ell)$ 
 begins with $\B_j$, while $D''$ must begin with a 0, since it starts with $\B_{2p}$ . Hence $D' \neq D''$, a contradiction.

The only remaining case is that  both the leftmost and the rightmost characters removed belong 
 to the same $\mathcal{\B}_{2p}^1 \X \text{\L}$-cluster. In particular such a cluster must contain three  
 $\mathcal{\B}^1_{2p} \X \text{\L}$-blocks, 
 denoted here with $[Block]_1, [Block]_2, [Block]_3$ such that 
 calling $D'$ the maximum suffix of the left $D$ that does not contain \L{} and $D''$ 
 the minimum prefix of the right $D$ that contains \L{}, the contraction must have the following structure:
 
 \begin{align} \label{samecluster-1}
    \T(\ell)=\dots \overbrace{[Blo}^{D_1'}\underbrace{\overbrace{ck]_1}^{D_1''}\,(\mathcal{\B}_{2p}^1 \X \text{\L})^* \,
    \overbrace{[Blo}^{D_2'}}_{D} \underbrace{\overbrace{ck]_2}^{D_2''} \, (\mathcal{\B}_{2p}^1 \X \text{\L})^* \,
    \overbrace{[Blo}^{D_3'}}_{D} \overbrace{ck]_3}^{D_3''}\dots
\end{align}
 
 with $D_1'' = D_2'' = D''$ and $D_2' = D_3' = D'.$
 
 We claim that it must also hold that $D_1' D_3''$ is a $\mathcal{\B}_{2p}^1 \X \text{\L}$-block, hence
 the contraction preserves the structure in (\ref{standardform}).
 
To see this, for each $i=1, 2, 3,$ let us write 
$$[Block]_i = \B_{2p} \B_{2p-1} \dots \B_2 \B_1^{(i)} \B_0 \X^{(i)} \text{\L} $$
 since they might only differ by the number of duplications in the chunks 
$\B_1^{(i)}$ and $\X^{(i)}.$ Since all chunks, $\B_{2p}, \dots, \B_0, \X$ have 
distinct numbers of runs, in order to have $D_1'' = D_2''$ (resp.\ $D_2' = D_3'$)  these strings must 
start (resp.\ end)  in the same chunk $\B_j$ or $\X.$

\noindent
{\em Case 1.} $D_1''$'s first character is a character of $\B_j$ for some $j \neq 1.$ 
Then, necessarily, the prefix of $D_3''$ up to the first \$ is equal to the prefix of $D_2''$ up to the first \$.
Therefore the contraction under consideration 
is equivalent to the alternative contraction where  the right $D$ (resp.\ left $D$) extends from 
the first character of $\B_j$ in $[Block]_1$ (resp.\  $[Block]_2$) to the last character of $\B_{j+1}$ in   
$[Block]_2$ (resp.\  $[Block]_3$). Hence, the contraction, coincides with the removal of a sequence of whole chunks, namely 
$\B_j \dots \B_0 \X\text{\L}(\mathcal{\B}^1_{2p} \X \text{\L})^* \B_{2p} \dots \B_{j+1}.$ 
It is easy to see that in this case the claim in item 1 of the proposition is true, as well as the claim in item 2, as 
subcase (b).

\noindent
{\em Case 2.} $D_1''$'s first character belongs to a chunk $\C_1 \in \{B^{(1)}, \X^{(1)}\}.$ 
Then there are  chunks $\C_1 = \C_1' \C_1'; \, \C_2 =  \C_2' \C_2''; \,  \C_3 = \C_3' \C_3''$ such that  for $\C \in \{\B_1, \X\}$ for each 
$i=1, 2,3,$ it holds that $\C \subseteq \C_i \subseteq \C^*$\footnote{We are abusing notation and for the sake of conciseness, we write $\X^* = \X^{dup}$} 
and we can factorize $[Block]_i$ as follows
$$[Block]_i = \overbrace{A'_i \hat{C}_i'}^{D_i'} \overbrace{\hat{C}_i'' A''_i}^{D_i''},$$
where $A_i'$ and $A_i''$ are sequences of chunks from $\hat{\Omega}.$
\remove{
By Fact \ref{mu-dup-modularity-2}, there is $C \in \{B_1, X\}$ and indices $i_1,  i_2,  i_3,$ and 
sets of indices $ I_1, I_2, I_3,$ such that
\begin{eqnarray*}
\C_1' &=& \C_1[1, i_1] = dup(\mu(C), I_1)[1, i_1] \\
\C_1'' &=& \C_1[i_1+1, |\C_1| ] = dup(\mu(C), I_1 )[i_1+1, |\C_1|] \\
\C_2' &=& \C_1[1, i_2] = dup(\mu(C), I_2)[1, i_2] \\
\C_2'' &=& \C_1[i_2+1, |\C_2| ] = dup(\mu(C), I_2)[i_2+1, |\C_2|] \\
\C_3' &=& \C_3[1, i_3] = dup(\mu(C), I_3)[1, i_3] \\
\C_3'' &=& \C_3[i_3+1, |\C_3| ] = dup(\mu(C), I_3)[i_3+1, |\C_3|] 
\end{eqnarray*} 

Hence, we must have $i_2 = i_3.$ 
}
%
From (\ref{samecluster-1}), using the above definitions, we have

 \begin{align} \label{samecluster-2}
    \T(\ell)=\dots \overbrace{A_1' \C_1'}^{D_1'}\underbrace{\overbrace{\C_1''A_1''}^{D_1''}\,(\mathcal{\B}_{2p}^1 \X \text{\L})^* \,
    \overbrace{A'_2 \C_2'}^{D_2'}}_{D} \underbrace{\overbrace{\C_2'' A_2''}^{D_2''} \, (\mathcal{\B}_{2p}^1 \X \text{\L})^* \,
    \overbrace{A_3'\C_3'}^{D_3'}}_{D} \overbrace{\C''_3 A_3''}^{D_3''}\dots
\end{align}
It holds that  $\C_2'' = \C_1'',$ and $\C_2' = \C_3'$ and $A_1'' = A_2'', \, A_2' = A_3'.$

And, after the removal of the second $D = \C_2'' A_2'' (\mathcal{\B}_{2p}^1 \X \text{\L})^* A_3' \C_3'$, 
the result of the contraction is the string

 \begin{align} \label{samecluster-3}
    \T(\ell+1)=\dots A_1' \C_1 A_1'' \,(\mathcal{\B}_{2p}^1 \X \text{\L})^* \,
    A'_2 \C_2' \C''_3 A_3'' \dots
\end{align}
Since $\C_2' \C''_3 = C_3,$ it follows that
  $A'_2 \C_2' \C''_3 A_3'' = A'_2 \C_3 A_3''$ is a 
$\mathcal{\B}^1_{2p} \X \text{\L}$-block. Hence $\T(\ell+1)$ also has the form (\ref{standardform}), 
proving the item 1. in the proposition's statement.

Moreover, it is easy to check that the item 2. of the proposition's statement is also satisfied as subcase (c).


\remove{    
\begin{align*}
    \T(\ell-1)=\dots (\mathcal{\B}_{2p}^1 \underbrace{\X \text{\L})(\mathcal{\B}_{2p}^1 \X \text{\L})^*
    (\overbrace{\mathcal{\B}_{2p}^1}^{D'}}_{D} \underbrace{\overbrace{\X \text{\L}}^{D''})(\mathcal{\B}_{2p}^1 \X \text{\L})^*
    (\overbrace{\mathcal{\B}_{2p}^1}^{D'}}_{D} \overbrace{\X \text{\L}}^{D''})\dots
\end{align*}
    
    Notice that removing the right $D$, the resulting string is still a $\mathcal{B}_{2p}^1 \mathcal{X} \text{\L}$-cluster because $D'D''=\mathcal{B}_{2p}^1 \mathcal{X} \text{\L}$, hence the structure of the claim is preserved.
 }

\item The leftmost character removed belongs to a $E_{i_j}(\ell)$ substring for some $j \in [h]$ 
and the rightmost character removed belongs to $E_{i_{j+k}}(\ell)$ for some integer $k$ such that $1 \leq k \leq h-j$. 
Assume that $k>1$, then we have that $\text{\L}E_{i_{j+1}}(\ell)$ is a substring of the right $D$. 
This implies that also the left $D$ must contain this substring, but this is impossible because the number of runs of each block $E_i(\ell)$ is different and 
each $\mathcal{\B}_{2p}^a \X \text{\L}$-block has more runs than $E_{i_{j+1}}(\ell)$ by construction. 
Therefore we know that the rightmost character removed belongs to $E_{i_{j+1}}(\ell)$. 
To summarize, so far we showed that if the leftmost character removed belongs to 
$E_{i_j}(\ell)$ and all the previous assumptions hold, the contraction looks like this:
    
\begin{align*}
   \T(\ell)=\underbrace{\dots\text{\L}E_{i_j}(\ell}_{D}\underbrace{)(\mathcal{\B}_{2p}^1 \X \text{\L})^* \overbrace{E_{i_{j+1}}(\ell}^{\mathclap{D'}}}_{D})\dots
 \end{align*}
    
   We now show that the leftmost character of $DD$ belongs either to the $\mathcal{\B}_{2p}^1 \X \text{\L}$-cluster 
    exactly before $E_{i_j}(\ell)$, or to the first block. 
    Assume that the leftmost character of $DD$ belongs to $E_{i_{j-1}}(\ell)$ (we know that it cannot 
    belong to any $E_{i_{j-k}}(\ell)$ with $k>1$). Then if the leftmost character removed is the first \$ of 
    $E_{i_j}(\ell) = \B_{i_j}\$\B_{i_j-1}\$\B_{i_j-2}\$\dots$ or any character further left, we know that the right $D$ must contain the substring $\$\B_{i_j-1}\$$, and therefore also the left $D$ must contain it. 
    This is impossible, since every substring $\B_j$ has a different number of runs, 
    $E_{i_j-1}(\ell) = \B_{i_j-1}\$\B_{i_j-2}\$\dots$ is a block, so it is preceded by an \L{}, 
    and every $E_{i_j-k}(\ell)$ with $k>1$ doesn't contain $\$\B_{i_j-1}\$$ as a substring by construction. 
    If the leftmost character removed is more to the right than the first \$ of $E_{i_j}(\ell)$, then $\text{\L}\B_{i_j}\$$ 
    is a substring of the left $D$, therefore also the right $D$ must contain this substring. 
    This means that there exists a string $E_{i_j+k}(\ell) = \B_{i_j+k}\$\B_{i_j+k-1}\$\dots$ with $k>0$ that contains 
    $\B_{i_j}\$$ as a prefix, but this is impossible by construction. Therefore we showed that the leftmost character of $DD$ belongs either to the $\mathcal{\B}_{2p}^1 \X \text{\L}$-cluster exactly before $E_{i_j}(\ell)$, or to the first block.
    
    We now prove that the result of a contraction like this is the same as removing the $E_{i_j}(\ell)$ string 
    (possibly together with some $\mathcal{\B}_{j}^a \X \text{\L}$-blocks). As before, we call $D'$ (resp. $D''$) the maximum suffix of the right (resp. left) $D$ that doesn't contain \L{}. The contraction now can be imagined as:
    
    \begin{align} \label{lastcase}
    \T(\ell)=\dots(\mathcal{\B}_{2p}^1 \underbrace{\X \text{\L})(\mathcal{\B}_{2p}^1 \X \text{\L})^*\overbrace{E_{i_j}(\ell}^{D''}}_{D}\underbrace{)(\mathcal{\B}_{2p}^1 \X \text{\L})^* \overbrace{E_{i_{j+1}}(\ell}^{\mathclap{D'}}}_{D})\dots
    \end{align}
    
    Notice that the leftmost character removed in $E_{i_j}(\ell) = \B_{i_j}\$\B_{i_j-1}\$\B_{i_j-2}\$\dots$ 
    must be more to the left than the first \$ of $E_{i_j}(\ell)$ for an argument similar to the one used before. 
    We know that $D'=D''$, otherwise $DD$ couldn't be a valid contraction. 
    Then if we imagine deleting the right $D$, we can see that the form of the resulting string 
    $\T(\ell+1)$ is equal to $\T(\ell)$ without $E_{i_j}(\ell)$ (together with  some adjacent $\mathcal{\B}_{j}^a \X \text{\L}$-blocks). 
    This implies that the form of the claim is preserved.
    
In fact, we can show that such a contraction is equivalent to a contraction that does not include any element of $E_{i_{j+1}}(\ell).$

To see this, let us denote by $\B_{i_j}'$ the prefix of $\B_{i_j}$ that is a suffix of the left $D$. Let $\B_{i_j}''$
be  the suffix of $\B_{i_j}$ 
which is a prefix of the right $D$. Recall that the first character removed is a character to the left of the first \$ in the block
$E_{i_j}(\ell) = \B_{i_j}\$ \B_{i_j-1}\$ \cdots \text{\L}.$ Since the contraction only removes some character from $\B_{i_j}$, we must have
 $i_j > 1.$ This is a consequence of the following claim whose proof is deferred below.
 \begin{Claim}
\label{claim:no2edges}
For every $\ell$, if the contraction $\T(\ell) \rightarrowtail  \T(\ell+1)$ removes characters from the substring 
$\$\B_{1}\$\B_{0}\$\X$ belonging to some block $E_{i_j}(\ell)$, then it  removes all the characters contained in 
$\$\B_{1}\$\B_{0}\$\X.$ 
\end{Claim}
 
Let $\B_{i_{j+1}}'$ be the prefix of $E_{i_{j+1}}(\ell) = \B_{i_{j+1}}\$ \B_{i_{j+1}-1} \dots \text{\L}$  coinciding with $D'.$ Hence, in particular, 
$\B_{i_{j+1}}' = \B_{i_j}'.$ Moreover, in both copies of $D$ between the prefix $\B_{i_j}''\$$ and the first $\text{\L}$ we must have 
the  same string $\B_{i_j-1}\$ \dots \$ \B_1' \$ \B_0\$ \X_{i_j}' \text{\L}$ (which is the suffix of $E_{i_j}(\ell)$ starting immediately after the first $\$$).
 
Then,  by making explicit the above substructures of the first block intersecting the right $D$ and the 
blocks $E_{i_j}(\ell)$ and $E_{i_{j+1}}(\ell)$, we have that (\ref{lastcase}) can be rewritten as
   \begin{align} \label{lastcase-2}
    \T(\ell)=\dots(\B_{2p}^1\$ \dots \$\B_{i_j}' \underbrace{\B_{i_j}''\$ \dots
     \X' \text{\L})(\mathcal{\B}_{2p}^1 \X \text{\L})^*\overbrace{\B_{i_j}'}^{D''}}_{D} \underbrace{\B_{i_j}'' \$ \dots \X' \text{\L} (\mathcal{\B}_{2p}^1 \X \text{\L})^* 
    \overbrace{\B_{i_j}'}^{D'}}_{D} \B_{i_{j+1}}'' \$ \dots \text{\L} \dots 
    \end{align}
    
  It is not hard to see that there is another possible contraction on $\T(\ell)$ which has exactly the same effect. In fact, we can 
  see that the suffix starting with $\B_{i_j}'$ in the first block shown in (\ref{lastcase-2}) is equal to $E_{i_j}(\ell)$. 
  Hence the $\T(\ell)$ can be actually factorized as
   \begin{align} \label{lastcase-3}
    \T(\ell)=\dots(\B_{2p}^1\$ \dots \$\underbrace{\overbrace{\B_{i_j}' \B_{i_j}''\$ \dots \X \text{\L}}^{E_{i_j}(\ell)})(\mathcal{\B}_{2p}^1 \X \text{\L})^*}_{\tilde{D}}\underbrace{E_{i_j}(\ell) (\mathcal{\B}_{2p}^1 \X \text{\L})^*}_{\tilde{D}} 
    E_{i_{j+1}}(\ell) \dots 
    \end{align}
 where the contraction that removes the right $\tilde{D}$ produces exactly the same string $\T(\ell+1)$ that is obtained from 
 removing $D$ in (\ref{lastcase}). 
 
 It is to be noticed that $\tilde{D}$ is a sequence of elements from $\tilde{\Omega}.$ This shows that item 2. in the 
 proposition's statement  is also 
 verified as subcase (b).
  
\end{enumerate}

We covered the cases in which the leftmost character removed belongs to the first block $ BX\text{\L}$, a substring $E'_j(\ell)$ for some 
$j>0$ or a $\mathcal{B}_{2p}^1 \mathcal{X} \text{\L}$-cluster, hence we exhausted every possibility and 
the proof of the induction step is complete.

\bigskip
\noindent
{\em Proof of Claim \ref{claim:no2edges}.} We argue by contradiction. 
Suppose that for some integer $\ell \geq 0,$  there exists a contraction 
$DD$ on $\T(\ell)$ that removes some but not all of the characters of the subsequence 
$\$\B_{01}\$\B_{0}\$\X$ that belongs to $E_{i_j}(\ell)$ for some $j>0$. 
By assumption the string $\T(\ell)$ 
has the form in (\ref{standardform}), and $D$ must contain at least 
one \L{} character. 
Thus, the left $D$ must contain the substring $\text{\L{}}\B_{i_j}\$$ that belongs to $E_{i_j}(\ell)$, 
but the right $D$ cannot contain this substring since it is unique in $\T(\ell)$. 
Therefore there cannot exist a contraction that removes some but not all of the characters 
of $\$\B_{1}\$\B_{0}\$\X$. \qed

\end{proof}

Using  the factorization of $\T(\ell)$ into chuncks, given by (\ref{standardform}), and the 
definition of  $\mu,$ we have that 
the string $T(\ell) = \mu^{-1}(\T(\ell))$ is a well defined sequence of chunks from $\Omega$, for each $\ell=0,1,\dots, k.$ 
In particular, we have $\mu^{-1}(\T(0)) = T$ and $\mu^{-1}(\T(k)) = \mu^{-1}(\S) = S.$

As a result of point 2. of the previous proposition we have the following.
\medskip

\noindent
\begin{Claim} \label{claim:equivalence2}
For each $\ell = 0, 1, \dots, k-1,$ 
if there is a contraction $\T(\ell)) \rightarrowtail \T(\ell+1))$ then  $T(\ell)) \rightarrowtail T(\ell+1))$ holds too.
\end{Claim}

This claim implies that there is a sequence of contractions
$$T = \mu^{-1}(\T) = \mu^{-1}(\T(0))  \rightarrowtail T(1) \rightarrowtail \cdots \rightarrowtail T(k) = \S$$
that proves the ``only if''  part of the statement of the Lemma.

\medskip
\noindent
{\em Proof of Claim \ref{claim:equivalence2}.}  Fix $0 \leq \ell \leq k-1.$ Let $\T(\ell) = \A \D \D \A' \rightarrowtail \A \D \A' = \T(\ell+1).$
We argue by cases according to three possibilities singled out by Proposition \ref{proposition:form}, item 2. 

\noindent
{\em Case 1.}  $\D$ is a single character within one chunk $\C$ of $\T(\ell).$ Therefore,  there must exist, 
two possibly empty sequences of chunks, $\T'$ and $\T''$ together with  an index $i^*$ and a chuck $\C' \in \hat{\Omega}$ 
such that $\C = dup(\C', i^*)$ 
$\T(\ell) = \T' \C \T''$ and $T(\ell+1) = T' \C' T''.$

Since $\mu$ is one-one, we have that there are chunks $C, C' \in \Omega$ such that $\C' = \mu(C'),$ and 
$\C = \mu(C) = dup(\mu(C'), i^*) = \mu(dup(C', i^*)).$ Then we have that 
\begin{eqnarray*}
T(\ell) &=& \mu^{-1}(\T') \mu^{-1}(\C)  \mu^{-1}(\T'') =  \mu^{-1}(\T') dup(C', i^*) \mu^{-1}(\T''),\\ 
\mu^{-1}(\T(\ell+1)) &=& \mu^{-1}(\T') C'  \mu^{-1}(T''),
\end{eqnarray*}
from which it is easier to see that there is a
contraction from $T(\ell))$ to $T(\ell+1).$ 

\medskip
 \noindent
{\em Case 2.}  $\T(\ell) = \A \D \D \A' \rightarrowtail \A \D \A' = \T(\ell+1)$ for some $\D$ being a sequence of whole chunks from $\hat{\Omega}.$
Then, also $\A$ and $\A'$ are sequences of whole chunks from $\hat{\Omega}.$ Hence, let $A = \mu^{-1}(\A), A' = \mu^{-1}(\A), 
D = \mu^{-1}(\D).$  We have
$T(\ell) = \mu^{-1}(\T(\ell)) = A D D  A' $ and $\mu(T(\ell+1))= A D A'$ which shows that 
there is a contraction from $T(\ell)$ to $T(\ell+1)$ which acts on the repeat $D D.$ 

\medskip
 \noindent
{\em Case 3.} there are  chunks $\C_1 = \C_1' \C_1''; \, \C_2 =  \C_2' \C_2''; \,  \C_3 = \C_3' \C_3''$ such that  for some $\C \in \{\B_0, \B_1, \X\}$ for each 
$i=1, 2,3,$ it holds that $\C \subseteq \C_i \subseteq \C^*$ and
 $\D = \C_1'' \D' \C_2' = \C_2'' \D' \C_3',$ and $C_1'$ is a suffix of $\A$ and $\C_3''$ is 
a prefix of $\A'.$

Let $C \in \{B_0, B_1, X\},$ such that $\C = \mu(C).$ Then, there are indices $i_1,  i_2, i_3$ and 
sets of indices $ I_1, I_2, I_3,$ such that
\begin{eqnarray*}
\C_1' &=& \C_1[1, i_1] = dup(\mu(C), I_1)[1, i_1] \\
\C_1'' &=& \C_1[i_1+1, |\C_1| ] = dup(\mu(C), I_1 )[i_1+1, |\C_1|] \\
\C_2' &=& \C_1[1, i_2] = dup(\mu(C), I_2)[1, i_2] \\
\C_2'' &=& \C_1[i_2+1, |\C_2| ] = dup(\mu(C), I_2)[i_2+1, |\C_2|] \\
\C_3' &=& \C_3[1, i_3] = dup(\mu(C), I_3)[1, i_3] \\
\C_3'' &=& \C_3[i_3+1, |\C_3| ] = dup(\mu(C), I_3)[i_3+1, |\C_3|], 
\end{eqnarray*} 
with  $\C_2'' = \C_1'',$ and $\C_2' = \C_3',$ hence $i_2 = i_3.$

Note that, for each $j=1,2,3,$ we can assume that $i_j \in [|C|+|I_j|].$ In fact, if this is not the case, since  
$$dup(\mu(C), I_j)[|C|+|I_j|, |\mu(C)|+|I_j] =  \mu(C)[|C|+1, |\mu(C)|],$$ we 
we can set $i_j =  |C|+|I_j|$ and still preserve the equalities $\C_2'' = \C_1'',$ and $\C_2' = \C_3'$ (this is equivalent to a 
``shift'' o the left of the window that selects the characters in $\D$,  without changing the result of the contraction).

Therefore, there are sequences of chunks $\T', \T'', \A_1, \A_2,$ such that we can factorize 
$\T(\ell)$ as follows (where  $\A = \T' \A_1 \C'_1,$ and $\A' = \C''_3 \A_2 \T''$):  

 \begin{align} \label{case-c-1}
    \T(\ell)=\T'  \A_1 \C_1' \underbrace{\C_1'' \D' \C_2'}_{\D} \underbrace{\C_2'' \D' \C_3'}_{\D} \C''_3 \A_2 \T''
\end{align}
Since  $\C_2'' = \C_1'',$  
the result of the contraction is the string

 \begin{align} \label{case-c-2}
    \T(\ell+1)=\T' \A_1 \C_1 \D' \C_2' \C''_3 \A_2 \T'' = \T' \A_1 \C_1 \D' \C'_3 \A_2 \T'' 
\end{align}
where, in the last equality we used $\C_2' \C''_3 = \C_3$ since $\C_2' = \C_3'.$. This also implies 
\begin{multline} \label{fromfact1}
dup(\mu(C), I_3)[1, |\C_3|] = \C_3 = \C_2'  \C_3'' = dup(\mu(C), I_2)[1, i_2]  dup(\mu(C), I_3)[i_3+1, |\C_3|]
\end{multline}

For $i=1, 2, 3,$ let $C_i = \mu^{-1}(\C_i).$ Moreover, let $T' = \mu^{-1}(\T'), \,
T'' = \mu^{-1}(\T''), D' = \mu^{-1}(\D') \, A_1 = \mu^{-1}(\A_1), A_2 = \mu^{-1}(\A_2).$

Then, we can factorize $T(\ell) = \mu^{-1}(\T(\ell))$ as follows:
\begin{align} \label{case-c-T}
    T(\ell)=T'  A_1 C_1[1,i_1] C_1[i_1+1, |C_1|] D' C_2[1,i_2] C_2[i_2+1, |C_2|] D' C_3[1,i_3] C[i_3+1, |C_3|] A_2 T''
\end{align}

By Fact \ref{mu-dup-modularity-2}, the equality 
\begin{multline}
 dup(\mu(C), I_1 )[i_1+1, |\C_1|] = \C_1[i_1+1, |\C_1| ] = \C_1'' \\
 = \C_2'' = \C_1[i_2+1, |\C_2| ] = dup(\mu(C), I_2)[i_2+1, |\C_2|]
 \end{multline}
 implies 
\begin{multline}
C_1[i_1+1, |C_1| ] = dup(C, I_1 )[i_1+1, |C_1|]  \\
 = dup(C, I_2)[i_2+1, |C_2|] = C_2[i_2+1, |C_2| ].
 \end{multline} 

Analogously, from $C_2' = C_3'$ and $i_2 = i_3$ we get 
\begin{multline}
C_2[1, i_2] = dup(C, I_2 )[1, i_2]  
 = dup(C, I_3)[1, i_3] = C_3[1, i_3 ].
 \end{multline}

$$dup(C, I_3) = dup(C, I_2)[1,\min\{i_2, |C|+|I'|\}]  dup(C, I_3)[i_3+1,|C|+|I_3|).$$

It follows that letting $D = C_1[i_1+1, |C_1|] D' C_2[1,i_2]$ we also have 
$D = C_2[i_2+1, |C_2|] D' C_3[1,i_3].$ Hence, 
$T(\ell) = T'  A_1 C_1[1,i_1] D D C_3[i_3+1, |C_3|] A_2 T''$ and we have the contraction 
\begin{multline} \label{laststeps}
T(\ell) = T'  A_1 C_1[1,i_1] D D C[i_3+1, |C_3|] A_2 T'' 
\rightarrowtail T'  A_1 C_1[1,i_1] D C[i_3+1, |C_3|] A_2 T''  \\
= T'  A_1 C_1[1,i_1] C_1[i_1+1, |C_1|] D' C_2[1,i_2] C_3[i_3+1, |C_3|] A_2 T''.
\end{multline}

Finally, using (\ref{fromfact1}) we have 

\begin{multline}
C_2[1,i_2] C_3[i_3+1, |C_3|] 
= dup(C, I_2)[1, i_2] dup(C, I_3)[i_3+1, |C|+|I_3|] \\
= dup(C, I_3)[1, |C_3|] = C_3
\end{multline}

It follows that the rightmost expression in (\ref{laststeps}) becomes 
$$T'  A_1 C_1 D' C_3 A_2 T''  = \mu^{-1}(\T(\ell+1)),$$ 
where the equality follows from (\ref{case-c-2}). We have then proved that the
contraction in (\ref{laststeps}) is the desired contraction from $T(\ell)$ to $T(\ell+1).$

%

\newpage
\section{The proof of Lemma \ref{lemma:function}}

\noindent
{\bf Lemma  \ref{lemma:function}.}
{\em Fix $q \in \{2,3,4,5\}.$ Let $S = 0^{l_1} 1^{l_2} \dots s_n^{l_n}$ 
and $T = 0^{l'_1} 1^{l'_2} \dots t_m^{l'_m}$ be purely alternating strings 
over the same $q$-ary alphabet $\Sigma = \{0,1,2,\dots, q-1\}.$ 
 Then, 
$S \Rightarrow^N_* T$ if and only if there exists 
a function $f:\{1, \dots, n-q+2\}\mapsto\{1, \dots, m-q+2\}$ 
such that:
\begin{enumerate}
\item $f(1) = 1$ and $f(n-q+2) = m-q+2$
\item $f(i) = j \implies s_i = t_j$ and  for each $u=0, \dots, q-2$ we have that $l_{i+u} \leq l'_{j+u}$
\item $f(i) = j$ and $f(i') = j'$ and $i<i' \implies j<j'$
    \item if $q =5,$ and  $f(i) = j$ and $f(i+1) = j' \neq j+1 \implies$ there exists a substring 
    $M$ in $T$ starting in a position $p$ such that $j \leq p \leq j'$ 
    with the form $M=s_{i+3}^{l'_p}, s_{i+4}^{l'_{p+1}}, \dots  s_{i+q}^{l'_{p+q-3}}, s_{i+q+1}^{l'_{p+q-2}}$ 
    such that for each $u=0, 1, \dots, q-3$ it holds that $l_{i+3+u} \leq l'_{p+u}$ and $l_{i+1} \leq l'_{p+q-2}.$
\end{enumerate}
}
\begin{proof} (sketch)
With reference to Definition \ref{defi:normal-duplications}, given a normal duplication, in case (i) (resp.\ (ii) ) we say that the duplication is {\em normal of type 1} (resp.\ {\em normal of type 2}).

We first prove the ``{\bf only if}'' direction of the statement. We argue by induction on the number of normal duplications. Let $\ell$ be such number, i.e.,  
$S \Rightarrow^N S_1 \Rightarrow^N S_2 \Rightarrow^N \cdots \Rightarrow^N S_{\ell -1} \Rightarrow^N T.$
Let $f_0(i) = i$ for each $i=1, \dots, n-q+2.$

\medskip
Clearly, if $\ell = 0,$ then $S = T$ and  $f_0$ satisfies properties 1-4.

Assume now that $\ell > 0$ and that the claim is true for $\ell -1,$ i.e., there is a function $f_{\ell-1}$ associated to 
$S \Rightarrow_{\ell-1}^N S_{\ell-1}$ that satisfies properties 1-4.

\noindent
{\em Case 1.} If the duplication $S_{\ell-1} \Rightarrow^N T$ is of type $2$ then it is easy to see that 
 $f_{\ell} = f_{\ell-1}$ satisfies the claim.
 
 \noindent
{\em Case 2.} If the duplication $S_{\ell-1} \Rightarrow^N T$ is of type $1,$ let $j^*$ be the first run of $RLE(S_{\ell-1})$ 
involved in this normal duplication. Then, the function 
$$f_{\ell}(i) = \begin{cases} f_{\ell-1}(i) & f_{\ell-1}(i) \leq j^*\\
f_{\ell-1}(i) + q & f_{\ell-1}(i) > j^*
\end{cases}$$
satisfies properties 1-3.

For property 4., in the case $q = 5$, we split the analysis into two sub-cases, according to whether the gap between 
$f_{\ell}(i)$ and $f_{\ell}(i+1)$ is created by the last duplication or not. Let $i^*$ be such that $f_{\ell}(i^*) = j^*.$

{\em Sub-case (i).} $f_{\ell}(i+1) \neq f_{\ell}(i) + 1$ and also $f_{\ell-1}(i+1) \neq f_{\ell-1}(i) + 1.$ 
Then there is a substring $M$ in $S_{\ell-1}$ starting in a position $f_{\ell-1}(i) \leq p \leq 
f_{\ell-1}(i+1)$ satisfying the property. Since this substring involves $q-1$ runs of $S_{\ell-1}$ it is not eliminated by the duplication 
$S_{\ell-1} \Rightarrow^N T$ which can only shift it to the right (by exactly $q$ positions,  which, by definition of $f_{\ell}$, only happens if  
$f_{\ell}(i+1) = f_{\ell-1}(i+1)+q$). Hence $M$ is also a substring of $T$ and it appears in a position $p'$ between $f_{\ell}(i)$ and $f_{\ell}(i+1).$

{\em Sub-case (ii).} $i = i^*$ and $f_{\ell-1}(i+1) = f_{\ell-1}(i) + 1.$ By definition, we have $f_{\ell}(i+1) = f_{\ell}(i) + q+1 \neq f_{\ell}(i)+1.$
In this case, it is not hard to verify that with $p = f_{\ell}(i^*) + 3$ the property is also satisfied.

\bigskip
\noindent
We now show the ({\bf if}) direction. 

Assume there exists a function $f$ satisfying properties $1-4.$ 
Because of properties 1-3, we have that there exists $\ell_1 \leq \ell_2,$ and $0 \leq  r < q$ such that 
$n = \ell_1 \times q + r$ and $m = \ell_2 \times q + r.$  Without loss of generality, we can assume $r = 0.$  
Let $\Delta = \Delta(S, T) = |RLE(T)|/n - |RLE(S)|/n = \ell_2 - \ell_1.$  We argue by induction on $\Delta$. 

For $\Delta = 0$ it must be $f(i) = i,$ for each $i$, hence, there is a sequence of  normal duplications of type 2 that satisfies the claim. 

Assume now $\Delta > 0$ and that (induction hypothesis)  
the claim holds true for any pair of strings $S'$ and $T'$ for which 
$RLE(S') = \ell_1' \times q, RLE(T') = \ell_2' \times q$ and $\Delta' = \Delta(S', T') = \ell_2' - \ell_1' \leq \Delta-1.$
 
Let $i$ be the smallest integer such that $f(i) = j$ and $f(i+1) = j' > j+1.$---Such an $i$ must exist, otherwise, we are in the case 
$\Delta = 0.$---Hence,  the fact that the strings are purely alternating, together with $f$ satisfying property 2., implies that 
$f(i+1) = (j+1) + u \times q$ for some $u \geq 1.$

\medskip 

Consider $\hat{S}$ obtained from $S$ by 
a normal duplication that starts on the $i$th run of $S$. Then 
$$\hat{S} = \dots, 
\underset{i}{s_i^{l_i}} \, \underset{i+1}{s_{i+1}^{l_{i+1}}} \; \dots, 
\underset{i+q-1}{s_{i+q-1}^{1}} \;  \underset{i+q}{s_{i+q}^1} \; 
\underset{i+q+1}{s_{i+1}^{l_{i+1}}}   \; 
\underset{i+q+2}{s_{i+2}^{l_{i+2}}} \;  \dots  
\underset{i+2q-1}{s_{i+q-1}^{l_{i+q-1}}} \; \dots,$$
where the number underneath each component specifies its run index.

Note that, by property 2., because of $f(i) = j,$ 
we have that $l_i \leq l'_j, \, l_{i+1} \leq l'_{j+1}, \,
l_{i+2} \leq l'_{j+2}, \, \dots, l_{i+q-2} \leq l'_{j+q-2}.$ Hence, letting $\hat{l}_k$ be the 
length of the $k$th run of $\hat{S},$ we also have 
$\hat{l_{i}} \leq l'_j, \, \hat{l}_{i+1} \leq l'_{j+1}, \,
\hat{l}_{i+2} \leq l'_{j+2}, \, \dots, \hat{l}_{i+q-2} \leq l'_{j+q-2}.$ Moreover we have
$1 = \hat{l}_{i+q-1} \leq l'_{j+q-1},$ and $1 = \hat{l}_{i+q} \leq l'_{j+q}.$

Finally, we have
$1 = \hat{l}_{i+q-1} \leq l'_{f(i+1)-2},$ and $1 = \hat{l}_{i+q} \leq l'_{f(i+1)-2}.$ And, by the hypothesis on $f(i+1)$ it holds that
\begin{equation} \label{properties-2}
\hat{l}_{i+q+u} = l_{i+u} \leq l'_{f(i+1)+u}, \mbox{ for } u= 1, \dots, q-1. 
\end{equation} 

Therefore, for $q < 5$, we have that the function 
$$\hat{f}(k) = \begin{cases} f(k) & k \leq i\\
f(i) + u & 1 \leq u \leq \min\{2, q-2\}, k = i+u,\\  
f(i+1)-2 &   k = i+q-1 \\
f(i+1)-1 &   k = i+q\\
f(k-q) & k \geq i+q+1,
\end{cases}
$$
satisfies properties $1-3$ for $\hat{S}$ and $T.$ 
Since $\Delta(\hat{S}, T) < \Delta, $ by induction hypothesis there is a sequence of normal duplications $\hat{S} \Rightarrow_* T$, 
that combined with the above normal duplication $S \Rightarrow^N \hat{S}$ proves the claim.

\bigskip

\noindent
Let us now consider the case $q = 5$.
Let $p$ be the index satisfying 4. Clearly, it holds that  $p = j+3 + v \times q,$ for some $0 \leq v \leq u-1.$
By property 4, we have  
\begin{equation} \label{property4-p}
\hat{l}_{i+3+u} \leq l_{i+3+u} \leq l'_{p+u} \mbox{ for } u = 0, 1, \dots, q-3 \mbox{ and }\hat{l}_{i+q+1} = l_{i+1} \leq l'_{p+q-2}.
\end{equation}

We distinguish three subcases accoording to whether: (i) $p = j+3;$ (ii) $p = j' - 3;$ (iii) $j+3 < p < j'-3.$
In each case, our argument will be to show that there is a sequence of normal duplication that can convert 
$S$ into a string $\hat{S}$ with more groups of runs---hence, such that $\Delta(\hat{S}, T) < \Delta$---for which there exists a function $\hat{f}$ satisfying properties 
1.-4. (with respect to $\hat{S}$ and $T$ and such that 
$\hat{f}(k) = k$ for $k \leq i+1$. This will imply, by induction hypothesis that $\hat{S} \Rightarrow^N_* T,$ hence proving our claim.

For the sake of conciseness, in the following argument we will denote by $[a]^{A}$  a normal duplication on a string $A$ that acts on the runs of $A$ from the 
$a$th one to the $(a+q-1)$th one.   
Recall that $i$ is the smallest integer such that $f(i) = j$ and $f(i+1) = j' > j+1.$

\noindent
{\em Subcase 1.} $p = j+3.$  Consider the duplication $[i]^S$  and indicate with $\hat{S}^{(1)}$ the resulting string. 

If $f(i+1) = i+q+1$ we have that  $\hat{S}^{(1)}$ satisfies our desiderata, i.e., indicating with $\hat{l}^{(1)}_j,$ the length of the 
$j$th run of $\hat{S}^{(1)}$ we have $\hat{l}^{(1)}_k \leq l'_k$ for each $k =1, \dots f(i+1).$ Hence, 
the function  
\begin{equation} \label{f:shat}
\hat{f}(k) = \begin{cases} k & k \leq f(i+1) \\ f(k-i-1) & k > i+1. \end{cases}
\end{equation}  satisfies 1.-4. w.r.t. $\hat{S}^{(1)}$ and $T$
and the desired result follows by induction. 

If $f(i+1) = i+1 + u \times q$ for some $u > 1,$ we use first a duplication $[i+3]^{\hat{S}^{(1)}}$ and  then repeatedly for $w = 3, 4, \dots, u$
(always on the newly obtained  string $\hat{S}^{(w)}$)
the duplication $[i+q+1]^{\hat{S}^{(w)}}.$ The resulting string $\hat{S}^{(w)}$ has a long sequence of runs of size $1$ and at each iteration 
the sequence of runs $s_{i+1+q}^{l_{i+1}} \, s_{i+2+q}^{l_{i+2}} \, \dots, s_{i+2q}^{l_{i+q}}$ gets shifted to the right until 
the run  
$s_{i+1+q}^{l_{i+1}}$ becomes the $f(i+1)$th run in $\hat{S}^{(u)}$ and we have that the function in (\ref{f:shat}) satisfies 
1.-4. w.r.t. $\hat{S}^{(u)}$ and $T$
and the desired result follows by induction. The process can be visualized as follows (where we are assuming $s_i = 0$):

    \begin{align*}
    S = {}& \dots \underline{0^{l_i} \ 1^{l_{i+1}} \ 2^{l_{i+2}} \ 3^{l_{i+3}} \ 4^{l_{i+4}}} \dots\\
\Rightarrow {}& \dots 0^{l_i} \ 1^{l_{i+1}} \ 2^{l_{i+2}} \ \underbrace{3^{l_{i+3}} \ 4^1 \ 0^1 \ 1^{l_{i+1}}}_{M} \ 2^{l_{i+2}} \ 3^{l_{i+3}} \ 4^{l_{i+4}}\dots  =  \hat{S}^{(1)}\\
     \hat{S}^{(1)} {}& \dots0^{l_i} \ 1^{l_{i+1}} \ 2^{l_{i+2}} \ \underline{3^{l_{i+3}} \ 4^1 \ 0^1 \ 1^{l_{i+1}} \ 2^{l_{i+2}}} \ 3^{l_{i+3}} \ 4^{l_{i+4}} \dots \\
    \Rightarrow {}& \dots 0^{l_i} \ 1^{l_{i+1}} \ 2^{l_{i+2}} \ \underbrace{3^{l_{i+3}} \ 4^1 \ 0^1 \ 1^{l_{i+1}}}_{M} \ 2^1 \ 3^1 \ 4^1 \ 0^1 \ 1^{l_{i+1}} \ 2^{l_{i+2}} \ 3^{l_{i+3}} \ 4^{l_{i+4}} \dots  = \hat{S}^{(2)},
    \end{align*}

where $M$ denotes a sequence of runs that satisfies the property 4.

\medskip
\noindent
{\em Subcase 2.} $p = j'-3.$  Again we perform first the duplication $[i]^S$  and indicate with $\hat{S}^{(1)}$ the resulting string. We then perform 
duplication $[i+2]^{\hat{S}^{(1)}}.$ The process can be visualized as follows (where we are assuming $s_i = 0$):

    \begin{align*}
   \hat{S}^{(1)} {}& \dots 0^{l_i} \ 1^{l_{i+1}} \ \underline{2^{l_{i+2}} \ 3^{l_{i+3}} \ 4^1 \ 0^1 \ 1^{l_{i+1}}} \ 2^{l_{i+2}} \ 3^{l_{i+3}} \ 4^{l_{i+4}} \dots \\
    \Rightarrow^N {}& \dots 0^{l_i} \ 1^{l_{i+1}} \ 2^{l_{i+2}} \ 3^{l_{i+3}} \ 4^1 \ 0^1 \ 1^1 \ 2^1 \ \underbrace{3^{l_{i+3}} \ 4^1 \ 0^1 \ 1^{l_{i+1}}}_{M} \ 2^{l_{i+2}} \ 3^{l_{i+3}} \ 4^{l_{i+4}} \dots  = \hat{S}^{(2)},
    \end{align*}
    where $M$ denotes a sequence of runs that satisfies the property 4. in $\hat{S}^{(2)}.$
    
    Like before, we continue with $(j'-j)/5-2$ normal duplications 
    producing the sequence of strings $\hat{S}^{(2)} \Rightarrow^N \hat{S}^{(3)} \Rightarrow^N \dots \Rightarrow^N \hat{S}^{(u)}$ where for each 
    $w = 2, \dots, u-1,$ the duplication is $[i+3]^{\hat{S}^{(w)}},$ i.e., 
    on the runs $3^{l_{i+3}} 4^1 0^1 1^1 2^1$. As in the previous case, we have that  the function in (\ref{f:shat}) satisfies 
1.-4. w.r.t. the resulting string $\hat{S}^{(u)}$ and $T$
and the desired result follows by induction. 
    
\medskip
\noindent
{\em Subcase 3.} $j+3 < p < j'-3$. 
As before, we perform first the duplication $[i]^S$  and indicate with $\hat{S}^{(1)}$ the resulting string. We then perform 
duplication $[i+3]^{\hat{S}^{(1)}},$ and on the resulting string, denote by  $\hat{S}^{(2)},$ we use duplication $[i+2]^{\hat{S}^{(2)}}.$
The process can be visualized as follows (where we are assuming $s_i = 0$):
    
    \begin{align*}
    \hat{S}^{(1)} {}& \dots0^{l_i} \ 1^{l_{i+1}} \ 2^{l_{i+2}} \ \underline{3^{l_{i+3}} \ 4^1 \ 0^1 \ 1^{l_{i+1}} \ 2^{l_{i+2}}} \ 3^{l_{i+3}} \ 4^{l_{i+4}} \dots \\
    \Rightarrow^N {}& \dots 0^{l_i} \ 1^{l_{i+1}} \ \underline{2^{l_{i+2}} \ \overbrace{3^{l_{i+3}} \ 4^1 \ 0^1 \ 1^{l_{i+1}}}^{M}} \ 2^1 \ 3^1 \ 4^1 \ 0^1 \ 1^{l_{i+1}} \ 2^{l_{i+2}} \ 3^{l_{i+3}} \ 4^{l_{i+4}} \dots  = \hat{S}^{(2)} \\
    \Rightarrow^N {}& \dots 0^{l_i} \ 1^{l_{i+1}} \ 2^{l_{i+2}} \ 3^{l_{i+3}} \ 4^1 \ 0^1 \ 1^1 \ 2^1 \ \underbrace{3^{l_{i+3}} \ 4^1 \ 0^1 \ 1^{l_{i+1}}}_{M} \ 2^1 \ 3^1 \ 4^1 \ 0^1 \ 1^{l_{i+1}} \ 2^{l_{i+2}} \dots \hat{S}^{(3)},
    \end{align*}
    where, in each intermediate string, we have indicated with $M$ a sequence of runs that satisfies the property 4. 
    
    We can now apply $(j'-j)/5-3$ normal duplications, choosing the four runs with length one to the left or the right of $M$. In this way we can position the substring $M$ to match the one in $T$ that is guaranteed by the fourth property and also shift the run $s_{i+1}^{l_{i+1}}$ to become the 
    $f(i+1)$th run of the resulting string $\hat{S}^{(u)}.$ As before, this implies that  the function in (\ref{f:shat}) satisfies 
1.-4. w.r.t. the resulting string $\hat{S}^{(u)}$ and $T$
and the desired result follows by induction.

%
%
%
%
\end{proof}

\section{The proof of Claim \ref{claim:narymodule}} 

\noindent
{\bf Claim \ref{claim:narymodule}.} 
{\em Let $S$ and $T$ be $q$-ary purely alternating strings such that there exists a series of duplications $S \Rightarrow_* T$. Then for each duplication, if we call the duplicated substring $X$, we have that $|RLE(X)| \bmod q \leq 1$.}
\begin{proof}
Assume that in the series of duplications $S \Rightarrow_* T$ there exists a duplication $S' = AXB \Rightarrow AXXB = T'$ with $|RLE(X)| \bmod q > 1$. In this case, the string $XX$ contains in the middle a substring $ab$, in which $a$ and $b$ are not consecutive characters in $S$. This happens because $a$ corresponds to the last character of $X$ and $b$ is the first character of $X$ and $|RLE(X)| \bmod q > 1$, implying that $a \neq (b-1) \bmod q$. For this reason $T'$ can't be purely alternating, but since $T$ is, there must exist some other duplication that eliminates the adjacency $ab$. This is impossible, since no duplication eliminates adjacencies, but can only create new ones.
\end{proof}

\section{The Algorithm in the proof of Theorem \ref{theo:main-purely}}
Here we show the pseudocode of an algorithm that can be used to check the existence of a mapping from 
$RLE(S)$ to $RLE(T)$ that guarantees that $S \Rightarrow^N_* T.$  
\begin{algorithm}[h!]
\caption{}
\label{alg:narylinear}
\textbf{Input:} Two $q$-ary purely alternating strings $S$ and $T$  ($q \in \{2, 3, 4, 5\}$)\\
\textbf{Output:} $yes$ if and only if there exists a function $f$ satisfying 1-4 in Lemma \ref{lemma:function}.\\
Compute $RLE(S) = (s_1, l_1),  \dots, (s_n, l_n)$ and $RLE(T) = (t_1, l'_1), \dots, (t_m, l'_m)$\;
\eIf{ $s_1 \neq t_1$ or $s_n \neq t_m$ or there is $u \in [q-1]$ s.t.\  $l_u>l'_u$ or $l_{n-u+1}>l'_{m-u+1}$}
    {\Return $no$\;}
    {$f(1) = 1, \, f(n-q+2) = m-q+2;$} 
$i=2$, $j=2$; \\
\While{$i \leq n-q+2$ and $j \leq m-q+2$}
  {\While{$l_{i+q-2} < l'_{j+q-2}$ and $i \leq n-q+2$}
  {$f(i) = j; \, i = i+1; \, j = j+1;$}
  \If{$i \leq n-q+2$}
  {$p = f(i-1)+3$\;
  \While{there is $u \in \{0,1,\dots, q-3\}$ s.t.\  $l_{(i-1)+3+u}>l'_{p+u}$ or $l_{i} > l'_{p+q-2}$}
  {$p = p+q$\;
   \If{$p > m-q+3$}
   {\Return $no$\;}
   } 
   $j = p-2+q$\;
   \While{ there is $u \in \{0,1,\dots, q-2\}$ s.t.\  $l_{i+u}>l'_{j+u}$ and $j < m-q+2$}
   {$j=j+q$\;
   }
   \eIf{$j \leq m-q+2$}       
   {$f(i) = j; \, i = i+1; j = j+1$}
   {\Return $no$\;}
  }
  }
\eIf{$i=n-q+1$ and $j =m-q+1$}               
{\Return $yes$\;}
{\Return $no$\;}
\end{algorithm}

\end{document}